%% file: RR-8528.tex
\documentclass[twoside]{article}
\usepackage[a4paper]{geometry}
\usepackage[latin1]{inputenc} % ou \usepackage[utf8]{inputenc}
\usepackage[T1]{fontenc} % ou \usepackage[OT1]{fontenc}
\usepackage{RR}
\usepackage{hyperref}
\usepackage{makeidx}         % allows index generation
\usepackage{graphicx}        % standard LaTeX graphics tool
\usepackage{multicol}        % used for the two-column index
\usepackage[bottom]{footmisc}% places footnotes at page bottom
\usepackage{amssymb,amsthm}
\usepackage{delarray}
\usepackage{amsmath}
\usepackage{latexsym}
\usepackage{tikz}
\usepackage{array,colortbl}
\usepackage[all,cmtip]{xy}
\usepackage{diagxy}
\usetikzlibrary{arrows,automata,decorations.pathmorphing,backgrounds,positioning,fit,shadows,shapes,topaths,through,trees,petri}
\newtheorem{theorem}{Theorem}[section]
\newtheorem{example}[theorem]{Example}
\newtheorem{proposition}[theorem]{Proposition}
\newtheorem{corollary}[theorem]{Corollary}

\newtheorem{definition}[theorem]{Definition}

\newtheorem{recall}[theorem]{Recall}
\newcommand{\mathbox}[1]{\mbox{{\footnotesize \mbox{$ #1 $}}}}

\newcommand{\set}[1]{\left\{{#1}\right\}} 
\newcommand{\setof}[2]{\left\{{#1\:}\left|{\;#2}\right.\right\}}

\newcommand{\lsta}[4]{\mathbox{#1\overset{#2}{\underset{#3}{\longrightarrow}}#4}}
\newcommand{\gsta}[4]{\mathbox{#1\overset{#2}{\underset{#3}{\Longrightarrow}}#4}}

\newcommand{\netfir}[3]{\mbox{$#1[#2\negthinspace\,\rangle\,\negthinspace #3$}}

\definecolor{grisclair}{gray}{0.8}
\definecolor{gristresclair}{gray}{0.97}
\newcommand{\ExpleBox}[1]
{{\leavevmode\unskip\nobreak\hfil\penalty50\hskip.75cm%
    \hbox{} \nobreak\hfil $\Box$~\textit{Exple}~\ref{#1} \parfillskip=0pt
    \finalhyphendemerits=0
    \par
}}
\newcommand{\PropBox}[1]
{{\leavevmode\unskip\nobreak\hfil\penalty50\hskip.75cm%
    \hbox{} \nobreak\hfil $\Box$~\textit{Prop}.~\ref{#1} \parfillskip=0pt
    \finalhyphendemerits=0
    \par
}}
\newcommand{\RecallBox}[1]
{{\leavevmode\unskip\nobreak\hfil\penalty50\hskip.75cm%
    \hbox{} \nobreak\hfil $\Box$~\textit{Recall}~\ref{#1} \parfillskip=0pt
    \finalhyphendemerits=0
    \par
}}
\newcommand{\DefBox}[1]
{{\leavevmode\unskip\nobreak\hfil\penalty50\hskip.75cm%
    \hbox{} \nobreak\hfil $\Box$~\textit{Def}.~\ref{#1} \parfillskip=0pt
    \finalhyphendemerits=0
    \par
}}

\RRNo{8528}
\RRdate{May 2014}
\RRauthor{% les auteurs
Eric Badouel\thanks{Inria and Lirima}%
  \and
Lo\"{\i}c H\'elou\"{e}t\thanks{Inria}%
\and Georges-Edouard Kouamou\thanks{ENSP and Lirima} 
\and Christophe Morvan \thanks{Université Paris-Est and Inria}} 
\authorhead{Badouel \& al.}
\RRtitle{Une approche grammaticale de la gestion de dossiers dans un environement collaboratif distribu\'e}
\RRetitle{A Grammatical Approach to Data-centric Case Management in a Distributed Collaborative Environment}
\titlehead{A Grammatical Approach to Data-centric Case Management}
\RRresume{Nous introduisons un mod\`ele d\'eclaratif de syst\`eme de gestion de dossiers \`a base d'artefacts 
ainsi qu'un sch\'ema de distribution de ce mod\`ele. Chaque dossier a une structure arborescente dont les moeuds 
combinent donn\'ees et calculs. Chaque noeud est sous la responsabilit\'e d'un unique utilisateur et des r\`egles 
s\'emantiques r\'egissent les \'evolutions du document ainsi que le calcul de la valeur des attributs des noeuds.
Les utilisateurs communiquent de fa\c{c}on asynchrone par \'echange de messages, et sans m\'emoire partag\'ee, ce qui facilite 
la distribution du mod\`ele. 
}
\RRabstract{
This paper presents a purely declarative approach to artifact-centric 
case management systems, and a decentralization scheme for this model. 
Each case is presented as a tree-like structure;
nodes bear information that combines data and computations. 
Each node belongs to a given stakeholder, and semantic rules govern the evolution of the tree structure, as well as how data values derive from information 
stemming from the context of the node. % --as in a spreadsheet system--.    
Stakeholders communicate through asynchronous message passing without shared memory, enabling convenient distribution. }
\RRmotcle{documents actifs, artefacts, gestion de dossiers, grammaires attribu\'ees}
\RRkeyword{Active Documents, Business Artifacts, Case Management, 
 Attribute Grammars}
\RRprojet{SUMO}  % cas d'un seul projet
\RCRennes % Rennes - Bretagne Atlantique
\begin{document}
\makeRR   % cas d'un rapport de recherche
%% \makeRT % cas d'un rapport technique.
%% a partir d'ici, chacun fait comme il le souhaite
\section{Introduction}\label{sec:introduction}
\input{introduction}

\section{A Grammatical Approach to Case Management}\label{sec:GACM}
\input{GACM}

\section{The Syntax of Guarded Attribute Grammars}\label{sec:syntax}
\input{syntax}

\section{The Behaviour of Guarded Attribute Grammars}\label{sec:semantics}
\input{semantics}

\section{Some Examples}\label{sec:examples}
\input{examples}

\section{Distribution of a Guarded Attribute Grammar}\label{sec:distribution}
\input{distribution}

\section{Soundness}\label{sec:soundness}
\input{soundness}

\section{Conclusion}\label{sec:conclusion}
\input{conclusion}

\bibliographystyle{plain}
\bibliography{biblio}
\end{document}

%% file: introduction.tex
% Introduction
Traditional Case Management Systems rely on workflow models. The emphasis is put on the orchestration of activities involving humans (the \textit{stakeholders})  
and software systems, in order to achieve some global objective. 
In this context, stress is often put on control and coordination of the tasks required for the realization of a particular service.  
Such systems are usually modeled using centralized and state-based formalisms like automata, Petri nets or statecharts. They can also be directly specified with dedicated 
notations like BPEL~\cite{BPEL07} or BPMN \footnote{\tt www.bpmn.org}. 

A drawback of existing workflow formalisms is that {\em data} exchanged during the processing of a task play a secondary role when 
not simply ignored. 
However, data can be tightly connected with control flows and should not be overlooked. 
Actually, data contained in a request may influence its processing;  conversely different decisions during the treatment of a case may produce distinct output-values.

Similarly, stakeholders are frequently considered as second class citizens in workflow systems: they are modeled as plain resources, 
performing specific tasks for a particular case, like machines in assembly lines.
As a result, workflow systems are ideal to model fixed production schemes in manufactures or organizations, but can be
too rigid to model open architectures where  
the evolving rules and data require more flexibility. 

On the other hand, {\em data-centric workflow systems}, 
proposed by IBM~\cite{NigamC03,Hull08,DamaggioDV12}, 
put stress on the exchanged documents, 
the so-called {\em Business Artifacts}, 
also known as {\em business entities with lifecycles}.
An artifact is a document that conveys all the information  
concerning a particular case from its inception in the  
system until its completion. 
It contains all the relevant information about the entity together with a lifecycle that models its possible evolutions through the business process. 
Several variants presenting the life cycle of an artifact by an automaton, a Petri net~\cite{LohmannW10},  
or logical formulas depicting legal successors of a state~\cite{DamaggioDV12} have been proposed. 
However, even these variants remain state-based centralized models in which stakeholders are second class citizens.  

Recently, Guard-Stage-Milestone (GSM), a declarative model of the lifecycle of artifacts has been 
introduced~\cite{HullDMFGHHHLMNSV11,DamaggioHV13}. 
This model defines {\em Guards}, {\em Stages} and {\em Milestones} to control the enabling, enactment and completion of (possibly hierarchical) activities. 
The GSM lifecycle meta-model has been adopted as a basis of the OMG standard 
{\em Case Management Model and Notation} (CMMN). 
The GSM model allows for dynamic creation of subtasks (the {\em stages}), and handles data attributes. Furthermore, guards and milestones attached 
to stages provide declarative descriptions of tasks inception and termination.
However, interaction with users are modeled as incoming messages from the environment, or as events from low-level (atomic) stages. 
In this way, users do not contribute to the choice of a workflow for a process. The semantics of GSM models is given in terms of global snapshots. 
Events can be handled by all stages as soon as they are produced, and guard of a stage can refer to attributes of distant stages. 
Thus this model is not directly executable on a distributed architecture.

\textit{This paper introduces a distributed and declarative model for Case Management called \textit{Guarded Attribute Grammars} 
(GAG for short), which is both data-centric, user-driven, and provides a convenient way to handle distribution}. 
GAGs are collections of semantic rules that describe how to produce data from inputs provided by the environment. 
They are a variant of attribute grammars~\cite{Knuth68,Paakki95}. 
Their notation is  reminiscent of unification grammars, and is inspired by the work of Deransart and Maluszynski~\cite{DeransartMBook93} 
relating attribute grammars with definite clause programs. 

In this declarative model, the lifecycle of artifacts is left implicit. 
Artifacts under evaluation can be seen as incomplete structured documents, i.e., 
trees with \textit{open nodes} corresponding to parts of the document that remain to be completed. 
Each open node is attached intentional data, i.e., an expression for the piece of information to be substituted to the node. 
The evolution of an artifact is governed by the stakeholder's decisions 
(choosing a particular action amongst those which are enabled at a given moment, inputing data,...), 
and by application of the semantic rules to update artifacts (by refining open nodes). 
Cases reach completion when they do not contain open nodes. 
An artifact is thus a structured document with some active parts. 
This notion of \textit{active documents} is close to the model of Active  
XML introduced by Abiteboul et al.~\cite{AbiteboulBMMW02} 
which consists of semi-structured documents with embedded service calls.

The paper is organized as follows. 
After an informal introduction to our grammatical approach to case management in Section~\ref{sec:GACM} 
we introduce the model of \textit{guarded attribute grammars} that underpins the approach by presenting successively 
its syntax (Section~\ref{sec:syntax}), its behaviour (Section~\ref{sec:semantics}) and by giving some illustrative examples (Section~\ref{sec:examples}). 
The deployment of a guarded attribute grammar on a distributed asynchronous architecture is studied in Section~\ref{sec:distribution}. 
Soundness of guarded attribute grammars is touched upon in Section~\ref{sec:soundness}. 
An assessment of the model and future research directions are given in conclusion.

%% file: GACM.tex
This section introduces a grammatical notation for case management which will be described more formally in the subsequent two sections.

Throughout the paper, the term \textit{case} designates a concrete instance of a given business process. 
We use the editorial process of an academic journal as a running example to illustrate the various notions and notations. 
A case for this example is the editorial processing of a particular article submitted to the journal.

The case is handled by various actors involved in the process, called \textit{stakeholders}, 
namely the editor in chief, an associate editor and some referees.
We associate each case with a document, called an \textit{artifact}, 
that collects all the information related to the case from its inception in the process until its completion. 
When the case is closed this document constitutes a full history of all the decisions that led to its completion.

We interpret a case as a problem to be solved, that can be completed by refining it into sub-tasks using business rules. 
This notion of business rule can be modelled by a \textit{production}
$
 P : s_0 \leftarrow s_1 \cdots s_n
$ 
expressing that task $s_0$ can be reduced to subtasks $s_1$ to $s_n$. 
If several productions with the same left-hand side $s_0$ exist then the choice 
of a particular production corresponds to a decision made by some designated stakeholder.
For instance, there are two possible immediate outcomes for a submitted article: 
either it is validated by the editor in chief 
and it enters the evaluation process of the journal or it is invalidated because its topic or format is not adequate. 
This initial decision can be reflected by the two following productions:
\[
\begin{array}{l@{\;\;:\;\;}l}
\mathrm{validate} &  \mathbf{Proposed\_submission} \leftarrow  \mathbf{Submission}\\
\mathrm{invalidate} & \mathbf{Proposed\_submission} \leftarrow
\end{array}
\]
If $P$ is the unique production having $s_0$ in its left-hand side, then there is no real decision to make and such a rule is 
interpreted as a logical decomposition of the task $s_0$ into substasks $s_1$ to $s_n$. 
Such a production will be automatically triggered without human intervention.

Accordingly, we model an artifact as a tree whose nodes are sorted.
We write $X::s$ to indicate that node $X$ is of sort $s$. 
An artifact is given by a set of equations of the form $X=P(X_1,\ldots,X_n)$, stating that $X::s$ is a node labeled by production 
$P : s \leftarrow s_1 \cdots s_n$ and with successor nodes $X_1::s_1$ to $X_n::s_n$. 
In that case node $X$ is said to be a \textit{closed} node defined by equation $X=P(X_1,\ldots,X_n)$ (we henceforth assume that we do not have two equations 
with the same left-hand side).
A node $X::s$ defined by no equation (i.e. that appears only in the right hand side of an equation) is an \textit{open node}. It corresponds to a pending task $s$.

The lifecycle of an artifact is implicitly given by a set of productions:
\begin{enumerate}
\item The artifact initially associated with a case is reduced to a single open node. 
\item An open node $X$ of sort $s$ can be {\em refined} by choosing 
a production\\ $P: s \leftarrow s_1\ldots s_n$ that fits its sort.\\
\parbox[c]{0.50\linewidth}{
The open node $X$ becomes a closed node $X=P(X_1,\ldots,X_n)$ under the decision of applying production $P$ to it. 
In doing so the task $s$ associated with $X$ is replaced by $n$ subtasks $s_1$ to $s_n$ and new open nodes 
$X_1::s_1$ to $X_n::s_n$ are created accordingly.
}\hfill
\parbox[c]{0.45\linewidth}{\vspace{-2ex}
\begin{center}
\begin{small}
\begin{tikzpicture}[>=stealth',shorten >=1pt, auto,on grid,>=stealth',auto,bend angle=30]
                    %state/.style={circle,minimum size=6mm,inner sep=1pt,draw=black!40,very thick,fill=black!10}]

\begin{scope}[scale=.8]
\node[state,draw=black!25,thick,fill=black!10,label=80:$s$]  (x) at (0,2)   {\small $?$};
\node[state,draw=black!25,thick,fill=black!10,label=80:$s$]  (x0) at (4,2)   {\small $P$};
\node[state,draw=black!25,thick,fill=black!10,label=110:$s_1$]  (x1) at (2.5,0)  {\small $?$ };
\node[state,draw=black!25,thick,fill=black!10,label=80:$s_n$]  (x2) at (5.5,0) {\small $?$};
\draw[->,thick] (1,2) -- (3,2);

\path[-,color=black!40,very thick] (x0)  edge (x1)
                                         edge (x2);
                                         
\end{scope}
\end{tikzpicture}
\end{small}
\end{center}
}
\item The case has reached completion when its associated artifact is closed, i.e. it no longer contains open nodes.
\end{enumerate}

However, plain context-free grammars do not model the interactions and data exchanged between the various tasks associated with open nodes. 
To overcome this problem, we attach additional information to open nodes using {\em attributes}. 
Each sort $s\in S$ comes equipped with a set of {\em inherited} attributes and a set of {\em synthesized} attributes.  
Values of attributes are given by {\em terms} over a ranked alphabet.
Recall that such a term is either a variable or an expression of the form $c(t_1,\ldots,t_n)$ where $c$ is a symbol of rank $n$, 
and $t_1, \dots , t_n$ are terms. In particular a constant $c$, i.e. a symbol of rank $0$, will be identified with the term $c(\,)$.
We will denote by $var(t)$ the set of variables used in term $t$.

\begin{definition}[Forms]
\label{def:form}
A {\bf form} of sort $s$ is an expression $$F=s(t_1,\ldots,t_n)\langle u_1,\ldots,u_m\rangle$$ where 
$t_1,\ldots,t_n$ (respectively $u_1,\ldots,u_m$) are terms over a ranked alphabet ~---the alphabet of attribute's values---~ 
and a set of variables $\mathit{var}(F)$.  
Terms $t_1,\ldots,t_n$ give the values of the {\bf inherited attributes} and  
$u_1,\ldots,u_m$ the values of the {\bf synthesized attributes}) attached to form $F$. 
\DefBox{def:form}
\end{definition}

From now on, we consider productions where sorts are replaced by forms of the corresponding sorts. 
More precisely, a production is of the form
\begin{equation}
 \label{eq:production}
 \begin{array}{lcl}
s_0(p_1,\ldots,p_n)\langle u_1,\ldots, u_m\rangle &\leftarrow& 
s_1(t_1^{(1)},\ldots,t_{n_1}^{(1)})\langle y_1^{(1)},\ldots, y_{m_1}^{(1)}\rangle\\ & & \cdots \\
 & & s_k(t_1^{(k)},\ldots,t_{n_k}^{(k)})\langle y_1^{(k)},\ldots, y_{m_k}^{(k)}\rangle
\end{array}
\end{equation}
where the $p_i$'s, the $u_j$'s, and the $t_j^{(\ell)}$'s are terms and the $y_j^{(\ell)}$'s are variables. 
The forms in the right-hand side of a production are {\em service calls}, namely they are 
forms $F=s(t_1,\ldots,t_n)\langle y_1,\ldots,y_{m}\rangle$ 
where the synthesized positions are (distinct) variables $ y_1,\ldots,y_{m}$  
(i.e.,\ they are not instantiated). 
The rationale is that we invoke a service by filling in the inherited 
positions of the form (the entries) and by indicating the variables that expect to receive the results returned by the service (the subscriptions). 

Any open node is now attached to a service call.
The corresponding service is supposed to {\em (i)} construct the tree that will refine the open node and 
{\em (ii)} compute  the values of the  synthesized attributes (i.e., it should return the subscribed values). 
A service is enacted by applying productions. 
More precisely, a production such as the one given in formula~(\ref{eq:production}) 
can apply in an open node $X$  when its left-hand side matches with the service call $s_0(d_1,\ldots,d_n)\langle y_1,\ldots, y_m\rangle$
attached to node $X$. 
For that purpose the terms $p_i$'s are used as patterns that should match the corresponding data $d_i$'s.
When the production applies, new open nodes are created and they are respectively associated with the forms 
(service calls) in the right-hand side of the production. The values of  $u_j$'s are then returned to the corresponding 
variables $y_j$'s that had subscribed to these values. For instance applying production
\[
 P : s_0(a(x_1,x_2))\langle b(y'_1),y'_2\rangle \leftarrow 
 s_1(c(x_1))\langle y'_1\rangle \;\; s_2(x_2,y'_1)\langle y'_2\rangle
\]
 \parbox[c]{0.60\linewidth}{
to a node associated with service call $s_0(a(t_1,t_2))\langle y_1,y_2\rangle$ 
gives rise to the substitution $x_1=t_1$ and $x_2=t_2$. The two newly-created open nodes are respectively 
associated with the service calls $ s_1(c(t_1))\langle y'_1\rangle$ and 
$s_2(t_2,y'_1)\langle y'_2\rangle$ and the values $b(y'_1)$ and $y'_2$ are substituted 
to the variables $y_1$ and $y_2$ respectively. 
}\hfill
\parbox[c]{0.38\linewidth}{\vspace{-2ex}
\begin{center}
\begin{tikzpicture}[>=stealth',shorten >=1pt, auto,on grid,>=stealth',auto,bend angle=20,
                    every place/.style={minimum size=4.5mm,thick,draw=black!60,inner sep=.2pt},
                    every transition/.style={minimum width=5mm,minimum height=5mm,inner sep=2pt,draw=black!60,thick,fill=white},
                    state/.style={circle,minimum size=5mm,inner sep=1pt,draw=black!25,thick,fill=black!10}]

\begin{scope}[scale=.6]
\node[state,label=90:{\scriptsize $s_0$}]  (nx) at (0,2)   {\scriptsize $\mathrm{P}$};
\node[state,label=90:{\scriptsize $s_1$}]  (nx1) at (-1.8,-.4)  {\tiny $?$};
\node[state,label=90:{\scriptsize $s_2$}]  (nx2) at (1.8,-.4) {\tiny $?$};
\path[-,color=black!40,very thick] (nx)  edge (nx1)
                                         edge (nx2);

\node[place,draw=green!45,thick,fill=green!10]  (a) at (-1.7,2.8)  {\scriptsize $a$};

\node[place,draw=green!45,thick,fill=green!10,,label=270:{\scriptsize $x_2$}]  (y) at (-1.2,1.8)  {\tiny $?$};
\node[place,draw=green!45,thick,fill=green!10,label=270:{\scriptsize $x_1$}]  (x) at (-2.2,1.8)  {\tiny $?$};
\path[-,thick,color=black!40] (a)  edge (x);
\path[-,thick,color=black!40] (a)  edge (y);

\node[place,draw=red!45,thick,fill=red!10]  (syn1) at (1.2,2.8) {\scriptsize $b$};
\node[place,draw=red!45,thick,fill=red!10]  (z) at (1.2,1.8)  {\scriptsize $y'_1$};
\node[place,draw=red!45,thick,fill=red!10]  (syn2) at (2.2,1.8)  {\scriptsize $y'_2$};
\path[-,thick,color=black!40] (syn1)  edge (z);

\node[place,draw=green!45,thick,fill=green!10]  (inh1) at (-2.5,-1.1)  {\scriptsize $c$};
\node[place,draw=green!45,thick,fill=green!10]  (xx) at (-2.5,-2.1)  {\scriptsize $x_1$};
\node[place,draw=red!45,thick,fill=red!10,label=below:{\scriptsize $y'_1$}]  (zz) at (-1.1,-1.1) {\tiny $?$};
%\path[-] (nx1)  edge (inh1)
%                edge (zz);
\path[-,thick,color=black!40] (inh1)  edge (xx);

\node[place,draw=green!45,thick,fill=green!10]  (inh22) at (1.1,-1.1)  {\scriptsize $y'_1$};
\node[place,draw=green!45,thick,fill=green!10]  (inh21) at (0.3,-1.1)  {\scriptsize $x_2$};
\node[place,draw=red!45,thick,fill=red!10,label=below:{\scriptsize $y'_2$}]  (uu) at (2.5,-1.1) {\tiny $?$};
%\path[-] (nx2)  edge (inh22)
%                edge (uu);
               
\draw[->,color=black!40,thick] (x) edge[bend angle=30,bend right] (xx.west);
\draw[->,color=black!40,thick] (y) edge[bend angle=10,bend left] (inh21.north);
\draw[->,color=black!40,thick] (uu) edge[bend angle=10,bend right] (syn2);
\path[->,color=black!40,thick] (zz.80) edge[bend angle=60,bend left] (inh22.north);
\path[->,color=black!40,thick] (zz.85) edge[bend angle=10,bend left] (z.250);

\end{scope}
\end{tikzpicture}
\end{center}}

The precise definitions  are given in the next section. 
For the moment, let us illustrate the notations on our running example. 
A stakeholder has a specific \textbf{role} in the editorial process: he can be an author, the editor in chief, 
an associate editor or a referee. Each role is associated with 
a set of services and a set of productions explaining how each service is provided. 
For instance an associate editor 
provides the service $\mathbf{Submission}(\mathit{article})\langle \mathit{decision}\rangle$ 
consisting in returning an editorial decision about an article submitted to the journal. 
We emphasize the fact that production 
$\mathrm{MakeDecision}(\mathit{decision}) : \mathbf{Decide}(\mathit{report}_1,\mathit{report}_2) \langle \mathit{decision}\rangle \leftarrow$
has a parameter $decision$, that is used to enter new data in the case. Parameters are a convenient way to avoid specifing a production for each 
allowed parameter. A parametric production can be equivalently replaced by several non-parametric productions when its parameters range over a finite set of values. 
The corresponding productions are listed in 
Table~\ref{Table:assEditor}. 
\begin{table}[htb]
\caption{Acting as an associate Editor}
\mbox{}\vspace{-.5cm}\mbox{}
\[
\begin{array}{rcl}
\mathrm{DecideSubmission} &:& \mathbf{Submission}(\mathit{article}) \langle \mathit{decision}\rangle\leftarrow \\ && \quad
\mathbf{Evaluate}(\mathit{article}) \langle \mathit{report}_1\rangle \\ && \quad
\mathbf{Evaluate}(\mathit{article}) \langle \mathit{report}_2\rangle \\ && \quad
\mathbf{Decide}(\mathit{report}_1,\mathit{report}_2) \langle \mathit{decision}\rangle \\ 
\mathrm{MakeDecision}(\mathit{decision}) 
&:& \mathbf{Decide}(\mathit{report}_1,\mathit{report}_2) \langle \mathit{decision}\rangle \leftarrow\\
\mathrm{AskReview}(\mathit{reviewer}) &:& \mathbf{Evaluate}(\mathit{article}) \langle \mathit{report}\rangle\leftarrow \\ && \quad
\mathbf{WaitReport}(\mathit{answer},\mathit{article})\langle \mathit{report}\rangle\\ && \quad
\mathtt{Call}(\mathit{reviewer},\mathbf{ToReview}(\mathit{article})\langle \mathit{answer}\rangle)\\
\mathrm{CaseNo}\langle \mathit{msg}\rangle &:&
\mathbf{WaitReport}(\mathrm{No}(\mathit{msg}),\mathit{article})\langle \mathit{report}\rangle \leftarrow \\ && \quad
\mathbf{Evaluate}(\mathit{article})\langle \mathit{report}\rangle\\
\mathrm{CaseYes}\langle \mathit{msg}\rangle &:&
\mathbf{WaitReport}(\mathrm{Yes}(\mathit{msg},\mathit{report}),\mathit{article})\langle \mathit{report}\rangle \leftarrow \\ 
\end{array}
\]
\label{Table:assEditor}
\end{table}

The first two productions mean that an associate editor makes an editorial decision about a submitted paper on the basis of the evaluation reports 
produced by two different referees.  
He can ask a report from a reviewer through an invocation of the external service $\mathbf{ToReview}(\mathit{article})\langle \mathit{answer}\rangle$.
The productions that govern the actions of a reviewer are given in Table~\ref{Table:reviewer}. 

\begin{table}[htb]
\caption{Acting as a reviewer}
\mbox{}\vspace{-.7cm}\mbox{}
\[
\begin{array}{rcl}
\mathrm{Decline}(\mathit{msg}) &:& 
\mathbf{ToReview}(\mathit{article})\langle \mathrm{No}(\mathit{msg})\rangle \leftarrow \\
\mathrm{Accept}(\mathit{msg}) &:&
\mathbf{ToReview}(\mathit{article})\langle \mathrm{Yes}(\mathit{msg},\mathit{report})\rangle \leftarrow 
\mathbf{Review}(\mathit{article})\langle \mathit{report}\rangle\\
\mathrm{MakeReview}(\mathit{report})&:& 
\mathbf{Review}(\mathit{article})\langle \mathit{report}\rangle \leftarrow\\
\end{array}
\]
\label{Table:reviewer}
\end{table}

One can group the productions of Table~\ref{Table:assEditor} and Table~\ref{Table:reviewer} using an additional parameter $\mathit{reviewer}$ 
to make as many disjoint copies of the specification given in  Table~\ref{Table:reviewer} as there are individuals playing the role of 
a referee. The resulting set of productions (where call to external services have been eliminated) is given in Table~\ref{Table:editorialProcess}. 
\begin{table}[htb]
\caption{Making a decision on a submitted paper}
\mbox{}\vspace{-.7cm}\mbox{}
\[
\begin{array}{rcl}
\mathrm{DecideSubmission} &:& \mathbf{Submission}(\mathit{article}) \langle \mathit{decision}\rangle\leftarrow \\ && \quad
\mathbf{Evaluate}(\mathit{article}) \langle \mathit{report}_1\rangle \\ && \quad
\mathbf{Evaluate}(\mathit{article}) \langle \mathit{report}_2\rangle \\ && \quad
\mathbf{Decide}(\mathit{report}_1,\mathit{report}_2) \langle \mathit{decision}\rangle \\ 
\mathrm{MakeDecision}(\mathit{decision}) 
&:& \mathbf{Decide}(\mathit{report}_1,\mathit{report}_2) \langle \mathit{decision}\rangle \leftarrow\\
\mathrm{AskReview}(\mathit{reviewer}) &:& \mathbf{Evaluate}(\mathit{article}) \langle \mathit{report}\rangle\leftarrow \\ && \quad
\mathbf{WaitReport}(\mathit{answer},\mathit{article})\langle \mathit{report}\rangle\\ && \quad
\mathbf{ToReview}(\mathit{reviewer},\mathit{article})\langle \mathit{answer}\rangle\\
\mathrm{Decline}(\mathit{msg})\langle \mathit{reviewer}\rangle &:& 
\mathbf{ToReview}(\mathit{reviewer},\mathit{article})\langle \mathrm{No}(\mathit{msg})\rangle \leftarrow \\
\mathrm{Accept}(\mathit{msg})\langle \mathit{reviewer}\rangle &:&
\mathbf{ToReview}(\mathit{reviewer},\mathit{article})\langle \mathrm{Yes}(\mathit{msg},\mathit{report})\rangle \leftarrow \\ && \quad
\mathbf{Review}(\mathit{reviewer},\mathit{article})\langle \mathit{report}\rangle\\
\mathrm{MakeReview}(\mathit{report})\langle \mathit{reviewer}\rangle&:& 
\mathbf{Review}(\mathit{reviewer},\mathit{article})\langle \mathit{report}\rangle \leftarrow\\
\mathrm{CaseNo}\langle \mathit{msg}\rangle &:&
\mathbf{WaitReport}(\mathrm{No}(\mathit{msg}),\mathit{article})\langle \mathit{report}\rangle \leftarrow \\ && \quad
\mathbf{Evaluate}(\mathit{article})\langle \mathit{report}\rangle\\
\mathrm{CaseYes}\langle \mathit{msg}\rangle &:&
\mathbf{WaitReport}(\mathrm{Yes}(\mathit{msg},\mathit{report}),\mathit{article})\langle \mathit{report}\rangle \leftarrow \\ 
\end{array}
\]
%\end{small}
\label{Table:editorialProcess}
\end{table}
Similarly one has as many instances of the productions in Table~\ref{Table:assEditor} as there are associate editors in the editorial board. 
In the complete (flat) specification one should 
therefore add an additional parameter $\mathit{associateEditor}$ to distinguish between all associate editors. 
If the specification is large and contains many different roles the resulting global grammar can be quite complex.
Yet, it is still possible to build an equivalent monolithic grammar without external service calls. 

The above specification uses production schemes rather than plain productions. 
Therefore the actual productions of the grammar are instances of these productions schemes where specific values are substituted to the parameters.
Replacing all parameters by their possible values to obtain plain productions in a systematic way results in a {\em guarded attribute grammar} 
(defined in Section~\ref{sec:syntax}) with an infinite set of productions.
However, at least in the above example, the parameters of the productions correspond either to a specific role in the process or to some kind of data 
(a message, a report, a decision) whose precise value has no impact on the behavior  of the system. 
Therefore one can abstract this specification by identifying all individuals 
playing the same role and by representing each type of data by a corresponding constant so that one can obtain 
a finite guarded attribute grammar with the same behavior.

%% file: syntax.tex
Attribute grammars, introduced by Donald Knuth in the late sixties \cite{Knuth68}, 
have been instrumental in the development of syntax-directed transformations and compiler design. 
More recently this model has been revived for the specification of structured document's manipulations 
mainly in the context of web-based applications.
The expression \textit{grammareware}  has been coined in \cite{KlintLV05} to qualify the tools for the design 
and customization of grammars and grammar-dependent softwares. 
One such interesting tool is the UUAG system developped by Swierstra and his group.  
They relied on  purely functional implementations of  attribute grammars 
\cite{Johnsson87,SaraivaSK00,Backhouse02} to 
build a domain specific languages (DSL) as a set of functional combinators derived from the 
semantic rules of an attribute grammar \cite{SwierstraAS98,SaraivaSK00,SaraivaS03}. 
We intend to adapt this construction to the model of guarded attribute grammars introduced in this paper.

An Attribute grammar is obtained from an underlying grammar by associating each sort $s$ with a set $\mathit{Att}(s)$ of {\em attributes} 
~---which henceforth should exist for each node of the given sort---~ and by associating each production $P:s\leftarrow s_1\ldots s_n$ 
with semantic rules describing the functional dependencies between the attributes of a node labelled $P$ (hence of sort $s$) and 
the attributes of its successor nodes (of respective sorts $s_1$ to $s_n$). 

We use a non-standard notation for attribute grammars, inspired from \cite{DeransartM85,DeransartMBook93}. 
Let us introduce this notation on an example before proceeding to the formal definitions. 

\input{flattening1}

Guarded attribute grammars extend the traditional model of attribute grammars by allowing 
patterns rather that plain variables (as it was the case in the above example) to represent 
the inherited attributes in the left-hand side of a production.  
Patterns allow the semantic rules to process by case analysis 
based on the shape of some of the inherited attributes, and in this way to handle the interplay between the data (contained in the inherited attributes) 
and the control (the enabling of productions). 

\begin{definition}[Guarded Attribute Grammars]
\label{def:GAG}
Given a set of sorts $S$ with fixed inherited and synthesized attributes. 
A \textbf{guarded attribute grammar} is a set of productions 
$
 P:F_0\;\leftarrow\; F_1 \cdots F_k
$ 
where the $F_i::s_i$ are forms.  
The inherited attributes of left-hand side $F_0$ are called the \textbf{patterns} of the production. 
The values of synthesized attributes in the right-hand side are variables.
These occurrence of variables together with the variables occurring in the patterns are called the {\bf input occurrences} of variables.
We assume that each variable has {\bf at most one input occurrence}.
\DefBox{def:GAG}
\end{definition}

The well-formedness conditions of GAGs express that every output is defined in terms of the inputs. 
We will often refer to this correspondences as the \textbf{semantic rules}.
More precisely, the inputs are associated with (distinct) variables and the value of each output is given by a term using these variables. 

Each variable can have several occurrences. 
First it should appear once as an input and it may also appear in several occurrences within some output term.
The corresponding occurrence is respectively said to be  in an {\em input} or in an {\em output position}. 
One can define the following transformation on productions whose effect is to annotate each occurrence of a variable 
so that $x^?$ (respectively $x^!$) stands for an occurrence of $x$ in an input position (resp. in an output position). 
\[
 \begin{array}{r@{\;=\;}l}
  !(F_0\leftarrow F_1 \cdots F_k) & ?(F_0) \leftarrow !(F_1) \cdots !(F_k)\\
  ?(s(t_1,\ldots t_n)\langle u_1,\ldots u_m\rangle) & s(?(t_1),\ldots ?(t_n))\langle !(u_1),\ldots !(u_m)\rangle\\
  !(s(t_1,\ldots t_n)\langle u_1,\ldots u_m\rangle) & s(!(t_1),\ldots !(t_n))\langle ?(u_1),\ldots ?(u_m)\rangle\\
  ?(c(t_1,\ldots t_n)) & c(?(t_1),\ldots ?(t_n))\\
  !(c(t_1,\ldots t_n)) & c(!(t_1),\ldots !(t_n))\\
  ?(x) & x^{?} \\!(x) & x^{!}
 \end{array}
\]
The conditions stated in Definition~\ref{def:GAG} say that in the 
labelled version of a production each variable occurs at most once in an input position, i.e.,\ that 
$\set{?(F_0), !(F_1), \ldots, !(F_k)}$ is an admissible labelling of the set of forms in $P$ 
according to the following definition. 

\begin{definition}[Link Graph]\label{def:DG}
A labelling in $\set{?,!}$ of the variables $\mathit{var}(\mathcal{F})$ of a set of forms $\mathcal{F}$ 
is \textbf{admissible} if the labelled version of a form $F\in\mathcal{F}$ is given by either $!F$ or 
$?F$ and each variable has at most one occurrence labelled with $?$. 
The occurrence $x^?$ identifies the place where the value of variable $x$ is defined and the occurrences 
of $x^!$ identify the places where this value is used. 
The \textbf{link graph} associated with an admissible labelling of a set of forms $\mathcal{F}$ is  
the directed graph whose vertices are the occurrences 
of variables with an arc from $v_1$ to $v_2$ if these vertices are occurrences of a same variable $x$, labelled $?$ in 
$v_1$ and $!$ in $v_2$. This arc, depicted as follows, \\
\begin{center}
\begin{tikzpicture}[>=stealth',shorten >=1pt, auto,on grid,>=stealth',auto,bend angle=30,
                    every place/.style={minimum size=4.5mm,thick,draw=black!60,inner sep=.2pt}]

\begin{scope}[scale=.8]
\node[place,draw=black!35,thick,fill=black!5,label=150:{\small $x$}]  (source) at (-1,0)  {\small $?$};
\node[place,draw=black!35,thick,fill=black!5]  (target) at (1,0)  {\small $x$};
\path[->,thick,color=black!40] (source) edge (target);  
\end{scope}
\end{tikzpicture}
\end{center}
means that the value produced in the \textbf{source vertex} $v_1$ should be forwarded to the \textbf{target vertex} $v_2$. 
Such an arc is called a \textbf{data link}. 
\DefBox{def:DG}
\end{definition}

%% file: flattening1.tex
\begin{example}[Flattening of a binary tree]\label{exple:flattening}
Our first illustration is the classical example of the attribute grammar that computes the flattening of a binary tree, 
i.e.,\, the sequence of the leaves  read from left to right. The semantic rules are usually presented as shown in
Table~\ref{table:flattening}.
\begin{table}[htb]
\caption{Flattening of a binary tree}
\begin{tabular}{l}
\parbox[c]{0.40\linewidth}{
\centering
\begin{tikzpicture}[>=stealth',shorten >=1pt, auto,on grid,>=stealth',auto,bend angle=80,
                    every place/.style={minimum size=4.5mm,thick,draw=black!60,inner sep=.2pt},
                    every transition/.style={minimum width=5mm,minimum height=5mm,inner sep=2pt,draw=black!60,thick,fill=white},
                    state/.style={circle,minimum size=6mm,inner sep=1pt,draw=black!25,thick,fill=black!10}]

\begin{scope}[scale=.6]
\node[state,label=80:{\scriptsize $X::\mathit{root}$}]  (x) at (0,1.5)   {\tiny $\mathrm{Root}$};
\node[state]  (x1) at (0,0)  {\tiny $?$};
\path[-,color=black!40,very thick] (x)  edge (x1);

\node[place,draw=red!45,thick,fill=red!10]  (syn) at (.8,.7) {\tiny $x$};
\path[-] (x)  edge (syn);
               
\node[place,draw=green!45,thick,fill=green!10]  (inh1) at (-.80,-.80)  {\tiny $\mathrm{Nil}$};
\node[place,draw=red!45,thick,fill=red!10,label=300:{\tiny $x$}]  (syn1) at (.80,-.80) {\tiny $?$};
\path[-] (x1)  edge (inh1)
               edge (syn1);

\draw[->,color=black!40,thick] (syn1)--(syn);

\end{scope}
\end{tikzpicture}
}\hfill 
\parbox[c]{0.5\linewidth}{
$
\begin{array}{l}
 \mathrm{Root} : \langle X::\mathit{root}\rangle \;\leftarrow \;\langle X_1::\mathit{bin}\rangle\\
 \qquad\qquad
 \begin{array}{l@{\quad}l@{\,=\;}l}
 \mathbf{where} & X\cdot s & X_1\cdot s\\
               & X_1\cdot h & \mathrm{Nil}
 \end{array}
\end{array}
$}
\\ \hline
\parbox[c]{0.40\linewidth}{
\centering
\begin{tikzpicture}[>=stealth',shorten >=1pt, auto,on grid,>=stealth',auto,bend angle=80,
                    every place/.style={minimum size=4.5mm,thick,draw=black!60,inner sep=.2pt},
                    every transition/.style={minimum width=5mm,minimum height=5mm,inner sep=2pt,draw=black!60,thick,fill=white},
                    state/.style={circle,minimum size=6mm,inner sep=1pt,draw=black!25,thick,fill=black!10}]

\begin{scope}[scale=.6]
\node[state,label=80:{\scriptsize $X::\mathit{bin}$}]  (x) at (0,2)   {\tiny $\mathrm{Fork}$};
\node[state]  (x1) at (-1.5,0)  {\tiny $?$};
\node[state]  (x2) at (1.5,0) {\tiny $?$};
\path[-,color=black!40,very thick] (x)  edge (x1)
                                        edge (x2);

\node[place,draw=green!45,thick,fill=green!10,label=150:{\tiny $x$}]  (inh) at (-1.2,1.8)  {\tiny $?$};
\node[place,draw=red!45,thick,fill=red!10]  (syn) at (1.2,1.8) {\tiny $y$};
\path[-] (x)  edge (inh)
              edge (syn);
               
\node[place,draw=green!45,thick,fill=green!10]  (inh1) at (-2.25,-.75)  {\tiny $z$};
\node[place,draw=red!45,thick,fill=red!10,label=300:{\tiny $y$}]  (syn1) at (-.75,-.75) {\tiny $?$};
\path[-] (x1)  edge (inh1)
               edge (syn1);

\node[place,draw=green!45,thick,fill=green!10]  (inh2) at (.75,-.75)  {\tiny $x$};
\node[place,draw=red!45,thick,fill=red!10,label=300:{\tiny $z$}]  (syn2) at (2.25,-.75) {\tiny $?$};
\path[-] (x2)  edge (inh2)
               edge (syn2);
               
\draw[->,color=black!40,thick] (syn1)--(syn);
\draw[->,color=black!40,thick] (inh)--(inh2);
\path[->,color=black!40,thick] (syn2.80) edge[bend right] (inh1.100);
\end{scope}
\end{tikzpicture}
}\hfill 
\parbox[c]{0.5\linewidth}{
$
\begin{array}{l}
 \mathrm{Fork} : \langle X::\mathit{bin}\rangle \;\leftarrow \;\langle X_1::\mathit{bin}\rangle\;\;\langle X_2::bin\rangle\\
 \qquad\qquad
 \begin{array}{l@{\quad}l@{\,=\;}l}
\mathbf{where} & X\cdot s & X_1\cdot s\\
               & X_1\cdot h & X_2\cdot s\\
               & X_2\cdot h & X\cdot h
\end{array}
\end{array}
$}
\\ \hline
\parbox[c]{0.40\linewidth}{
\centering
\begin{tikzpicture}[>=stealth',shorten >=1pt, auto,on grid,>=stealth',auto,bend angle=30,
                    every place/.style={minimum size=4.5mm,thick,draw=black!60,inner sep=.2pt},
                    every transition/.style={minimum width=5mm,minimum height=5mm,inner sep=2pt,draw=black!60,thick,fill=white},]
                    state/.style={circle,minimum size=6mm,inner sep=1pt,draw=black!40,very thick,fill=black!10}]

\begin{scope}[scale=.8]
\node[state,draw=black!25,thick,fill=black!10,label=80:{\scriptsize $X::\mathit{bin}$}]  (x) at (0,2)   {\tiny $\mathrm{Leaf}_a$};
                                         
\node[place,draw=green!45,thick,fill=green!10,label=150:{\tiny $x$}]  (inh) at (-.6,1)  {\tiny $?$};
\node[place,draw=red!45,thick,fill=red!10,inner sep=.2pt]  (syn) at (.6,1)   {\tiny $\mathrm{Cons}_a$};
\node[place,draw=red!45,thick,fill=red!10,inner sep=.2pt]  (syn1) at (.6,.2)   {\tiny $x$};
\path[-] (x)  edge (inh)
              edge (syn);
\path[-,thick,color=black!40]  (syn) edge (syn1);           
\path[->,thick,color=black!40] (inh) edge[bend right] (syn1);           

\end{scope}
\end{tikzpicture}
}\hfill 
\parbox[c]{0.50\linewidth}{
$
\begin{array}{l}
 \mathrm{Leaf}_a : \langle X::bin\rangle \;\leftarrow \\
 \qquad\qquad
 \begin{array}{l@{\quad}ll}
\mathbf{where} & X\cdot s\,=\; \mathrm{Cons}_a (X\cdot h)
\end{array}
\end{array}
$}
\end{tabular}
\label{table:flattening}
\end{table}
The sort $\mathit{bin}$ of binary trees has two attributes: the inherited attribute $h$ contains an accumulating parameter and 
the synthesized attribute $s$ eventually contains the list of leaves of the tree appended to the accumulating parameter. 
Which we may write as $t\cdot s= \mathit{flatten}(t)+\!\!+ t\cdot h$, i.e., $t\cdot s = \mathit{flat}(t,t\cdot h)$ where 
$\mathit{flat}(t,h)=\mathit{flatten}(t)+\!\!+ h$. The semantics rules stem from the identities:
\[
 \begin{array}{r@{=}l}
  \mathit{flatten}(t) & \mathit{flat}(t,Nil)\\
  \mathit{flat}(\mathrm{Fork}(t_1,t_2),h) & \mathit{flat}(t_1,\mathit{flat}(t_2,h))\\
  \mathit{flat}(\mathrm{Leaf}_a,h) & \mathrm{Cons}_a(h)
 \end{array}
\]
We present the semantics rules of Table~\ref{table:flattening} using the following syntax:
\[
 \begin{array}{l@{\;\;:\;\;}r@{\;\leftarrow\;}l}
  \mathrm{Root} & \mathit{root}()\langle x\rangle &\mathit{bin}(Nil)\langle x\rangle\\
  \mathrm{Fork} & \mathit{bin}(x)\langle y\rangle &\mathit{bin}(z)\langle y\rangle\;\mathit{bin}(x)\langle z\rangle\\
  \mathrm{Leaf}_a & \mathit{bin}(x)\langle \mathrm{Cons}_a(x)\rangle &
 \end{array}
\]
The \textit{syntactic categories} of the grammar, also called its \textit{sorts}, namely \textit{root} and \textit{bin} 
are associated with their inherited attributes (given as a list of arguments: $(t_1,\ldots,t_n)$) and 
their synthesized attributes (the co-arguments:$\langle u_1,\ldots,u_m\rangle$). 
A variable $x$ is an \textit{input variable}, denoted as $x^?$, if it appears in an inherited attribute of the left-hand side 
or in a synthesized attribute of the right-hand side. It corresponds to a piece of information stemming 
respectively from the context of the node or from the subtree rooted at the corresponding successor node. These variables should be 
pairwise distinct. Symmetrically a variable is an \textit{output variable}, denoted as $x^!$, if it appears in a synthesized 
attribute of the left-hand side or in an inherited attribute of the right-hand side. It corresponds to values computed by the 
semantic rules and send respectively to the context of the node or the subtree rooted at the corresponding successor node.
Indeed, if we annotate the occurrences of variables with their polarity (input or output) one obtains:
\[
 \begin{array}{l@{\;\;:\;\;}r@{\;\leftarrow\;}l}
  \mathrm{Root} & \mathit{root}()\langle x^{!}\rangle &\mathit{bin}(Nil)\langle x^{?}\rangle\\
  \mathrm{Fork} & \mathit{bin}(x^?)\langle y^!\rangle &\mathit{bin}(z^!)\langle y^?\rangle\;\mathit{bin}(x^!)\langle z^?\rangle\\
  \mathrm{Leaf}_a & \mathit{bin}(x^?)\langle \mathrm{Cons}_a(x^!)\rangle &
 \end{array}
\]
And if we draw an arrow from the (unique) occurrence of $x^?$ to the (various) occurrences of $x^!$ for each variable $x$ 
to witness the data dependencies then the above rules correspond precisely to the three figures shown on the left-hand side of 
Table~\ref{table:flattening}. 
\ExpleBox{exple:flattening}
\end{example}

%% file: semantics.tex
Attribute grammars are applied to input abstract syntax trees.
These trees are usually produced by some parsing algorithm during a previous stage. 
The semantic rules are then used to decorate the node of the input tree by attribute values.
In our setting the generation of the tree and its evaluation using the semantic rules are 
intertwined since the input tree represents an artifact under construction.
An artifact is thus an incomplete abstract syntax tree which contains closed and open nodes.
A closed node is labelled by the production that was used to create it. 
An open node is associated with a form that contains all the needed information for its further refinements. 
The information attached to an open node consists of the sort of the node 
and the current value of its attributes. 
The synthesized attributes of an open node are undefined and 
are thus associated with variables.

\begin{definition}[Configuration of a Guarded Attribute grammar]\label{def:configuration}
A {\bf configuration} $\Gamma$ of a guarded attribute grammar is an $S$-sorted set of nodes $X\in\mathit{nodes}(\Gamma)$ 
each of which is associated with a defining equation in one of the following form where $\mathit{var}(\Gamma)$ 
is a set of variables associated with $\Gamma$\vspace{-1mm}:
\begin{description}
 \item[Closed node:] $X=P(X_1,\ldots,X_k)$ where $P:s\leftarrow s_1\ldots s_k$ is a production of the underlying 
 grammar and $X::s$, and $X_i::s_i$ for $1\leq i\leq k$. Production $P$ is the {\bf label} of node $X$ and nodes $X_1$ to 
 $X_n$ are its {\bf successor nodes}. 
 \item[Open node:] $X=s(t_1,\ldots,t_n)\langle x_1,\ldots,x_m\rangle$ where $X$ is of sort $s$ and
 $t_1,\ldots,t_k$ are terms with variables in $\mathit{var}(\Gamma)$ that represents the values 
 of the inherited attributes of $X$, and $x_1,\ldots,x_m$ are variables in $\mathit{var}(\Gamma)$ 
 associated with its synthesized \vspace{-1mm}attributes.
\end{description}
Each variable in  $\mathit{var}(\Gamma)$ occurs at most once in a synthesized position. 
Otherwise stated $!\Gamma=\setof{!F}{F\in\Gamma}$ is an admissible labelling of the set of forms 
occurring in $\Gamma$. 
\DefBox{def:configuration}
\end{definition}

In order to specify the effect of applying a production at a given node of a configuration (Definition~\ref{def:firingRuleAG}) we first recall 
some notions about substitutions.

\begin{recall}[on Substitutions]\label{recall:substitutions}
We identify a substitution $\sigma$ on a set of variables $\set{x_1,\ldots,x_k}$, called the \textbf{domain} of $\sigma$, 
with a system of equations $$\setof{x_i=\sigma(x_i)}{1\leq i\leq k}$$ 
The set of variables of $\sigma$, 
defined by $\mathrm{var}(\sigma)=\bigcup_{1\leq i\leq k}\mathrm{var}(\sigma(x_i))$, is disjoint from the domain of $\sigma$. 
Conversely a system of equations $\setof{x_i=t_i}{1\leq i\leq k}$ defines a substitution 
$\sigma$ with $\sigma(x_i)=t_i$ if it is in \textbf{solved form}, i.e.,\  none of the variables $x_i$ appears in some of the terms $t_j$. 
In order to transform a system of equations $E=\setof{x_i=t_i}{1\leq i\leq k}$ into an equivalent system $\setof{x_i=t'_j}{1\leq j\leq m}$ in solved form one can 
iteratively replace an occurrence of a variable $x_i$ in one of the right-hand side term $t_j$  by its definition $t_i$ 
until no variable $x_i$ occurs in some $t_j$. 
This process terminates when the relation $x_i\succ x_j\Leftrightarrow x_j\in \mathrm{var}(\sigma(x_i))$ is acyclic. 
One can easily verify that, under this assumption, the resulting 
system of equation $SF(E)=\setof{x_i=t'_i}{1\leq i\leq n}$ in solved  does not depend on the order in which the variables $x_i$ have 
been eliminated from the right-hand sides.  
When the above  condition is met we say that  the set of equations is \textbf{acyclic} and that it \textbf{defines} the substitution 
associated with its solved form. 
\RecallBox{recall:substitutions}
\end{recall}

The composition of two substitutions $\sigma,\sigma'$ is denoted by $\sigma\sigma'$ and defined by $\sigma\sigma'=\{x=t\sigma'|x=t\in\sigma\}$. 
Similarly, we let $\Gamma\sigma$ denote the configuration obtained from $\Gamma$ by replacing the defining equation $X=F$ of each open node $X$
by $X=F\sigma$.

We now define more precisely when a production is enabled at a given open node of a configuration and the effect of applying the production.
First note that variables of a production are formal parameters which scope is limited to the production.
They can injectively be renamed in order to avoid clashes with variables names appearing in a configuration. 
Therefore we shall always assume that the set of variables of a production $P$ is disjoint from the set of variables of a configuration 
$\Gamma$ when applying production $P$ at a node of $\Gamma$. 
As informally stated in the previous section, a production $P$ applies at an open node $X$ when its left-hand side $s(p_1,\ldots,p_n)\langle u_1,\ldots u_m\rangle$ 
matches with the definition $X=s(d_1,\ldots,d_n)\langle y_1,\ldots,y_m\rangle$, i.e., the service call attached to $X$ in $\Gamma$.

First, the patterns $p_i$ should match with the data $d_i$ according to the usual pattern matching given by the following inductive statements 
\[
 \begin{array}{l}
\mathbf{match}(c(p'_1,\ldots,p'_k),c'(d'_1,\ldots,d'_{k'})) \mbox{~with~} c\neq c' \mbox{~fails~}\\ 
\mathbf{match}(c(p'_1,\ldots,p'_k),c(d'_1,\ldots,d'_k))\;=\;\sum_{i=1}^{k}\mathbf{match}(p'_i,d'_i) \\
  \mathbf{match}(x,d) = \set{x=d}
 \end{array}
\]
where the sum $\sigma =\sum_{i=1}^{k} \sigma_i$ of substitutions $\sigma_i$ is defined and equal to $\bigcup_{i\in 1..k} \sigma_i$ when all substitutions $\sigma_i$ are defined 
and associated with disjoint sets of variables.
Note that since no variable occurs twice in the whole set of patterns $p_i$, the various substitutions $\mathbf{match}(p_i,d_i)$, when defined, 
are indeed concerned with disjoint sets of variables. 
Note also that $\mathbf{match}(c(),c()) = \emptyset$.

\begin{definition}\label{def:match}
A form  $F= s(p_1,\ldots,p_n)\langle u_1,\ldots u_m\rangle$ {\bf matches} with a service call  \\
$F'=s(d_1,\ldots,d_n)\langle y_1,\ldots,y_m\rangle$ (of the same sort) when
\begin{enumerate}
 \item the patterns $p_i$'s matches with the data $d_i$'s, defining a substitution $\sigma_{\mathit{in}}=\sum_{1\leq i\leq n}\mathbf{match}(t_i,d_i)$, 
 \item the set of equations $\setof{y_j=u_j\sigma_{\mathit{in}}}{1\leq j\leq m}$ is acyclic and defines a substitution $\sigma_{\mathit{out}}$.
\end{enumerate}
The resulting substitution $\sigma=\mathbf{match}(F,F')$ is given by $\sigma= \sigma_{\mathit{out}}\cup \sigma_{\mathit{in}}\sigma_{\mathit{out}}$.  
\DefBox{def:match}
\end{definition}

\begin{definition}[Applying a Production]
\label{def:firingRuleAG}
Let $P=F\leftarrow F_1\ldots F_k$  be a production, $\Gamma$ be a configuration, and $X$ be an open node with 
definition $X=s(d_1,\ldots,d_n)\langle y_1,\ldots,y_m\rangle$ in $\Gamma$.
We assume that $P$ and $\Gamma$ are defined over disjoint sets of variables. 
We say that $P$ is \textbf{enabled} in $X$ and write $\netfir{\Gamma}{P/X}{}$, if the left-hand side 
of $P$ matches with the definition of $X$. Then applying production $P$ at node $X$ transforms configuration $\Gamma$ 
into  $\Gamma'$, denoted as  $\netfir{\Gamma}{P/X}{\Gamma'}$, where:
\[
  \begin{array}{rcl}
   \Gamma' & = & \set{X=P(X_1,\ldots,X_k)}\qquad \mbox{where}\quad X_1,\ldots,X_k\quad\mbox{are new nodes added to $\Gamma'$}\\
           & \cup & \set{X_1=F_1\sigma, \ldots, X_k=F_k\sigma} \\
           & \cup & \setof{X'=F\sigma}{(X'=F)\in\Gamma\;\wedge\; X'\neq X} 
  \end{array}
\]
where $\sigma=\mathbf{match}(F,X)$.
\DefBox{def:firingRuleAG}
\end{definition}

Thus the first effect of applying production $P$ to an open node $X$ is that $X$ becomes a closed node with label $P$ and successor 
nodes $X_1$ to $X_k$. The latter are new nodes added to $\Gamma'$. 
They are associated respectively with the instances of the $k$ forms in the right-hand side of $P$ obtained by applying substitution $\sigma$ to these forms.  
The definitions of the other nodes of $\Gamma$ are updated using substitution $\sigma$ (or equivalently $\sigma_{\mathit{out}}$). 
This update has no effect on the closed nodes because their 
defining equations in $\mathit{\Gamma}$ contain no variable. 

We conclude this section with two results %that serve to 
justifying Definition~\ref{def:firingRuleAG}.
Namely, Proposition~\ref{prop:firingruleAG} states that if $P$ is a production enabled in a node $X_0$ of a configuration $\Gamma$ 
with $\netfir{\Gamma}{P/X_0}{\Gamma'}$ then $\Gamma'$ is a configuration, i.e.,\ applying $P$ cannot create a variable with several 
input occurrences. And Proposition~\ref{prop:correctness} shows that the substitution $\sigma = \mathbf{match}(F,X)$  resulting 
from the matching of the left-hand side $F=s(p_1,\ldots,p_n)\langle u_1,\ldots, u_m\rangle$ of a production $P$  with the 
definition $X=s(d_1,\ldots,d_n)\langle y_1,\ldots,y_m\rangle$ of an open node $X$ is the most general unifier  of the set of equations 
$\setof{p_i=d_i}{1\leq i\leq n}\cup\setof{y_j=u_j}{1\leq j \leq m}$.

\begin{proposition}\label{prop:firingruleAG}
If production $P$ is enabled in an open node $X_0$ of a configuration $\Gamma$ and $\netfir{\Gamma}{P/X_0}{\Gamma'}$ 
then $\Gamma'$ is a configuration.
\end{proposition}
\begin{proof}
Let $P=F\leftarrow F_1\ldots F_k$ with left-hand side $F=s(p_1,\ldots,p_n)\langle u_1,\ldots u_m\rangle$ and 
$X_0=s(d_1,\ldots,d_n)\langle y_1,\ldots,y_m\rangle$ be the defining equation of $X_0$ in $\Gamma$. 
Since the values of synthesized attributes in the forms $F_1,\ldots, F_k$  are variables 
(by Definition~\ref{def:GAG}) and since these variables are unaffected by substitution $\sigma_{\mathit{in}}$ the 
synthesized attribute in the resulting forms $F_j\sigma_{\mathit{in}}$ are variables.
The substitutions $\sigma_{\mathit{in}}$ and $\sigma_{\mathit{out}}$ substitute terms to the variables $x_1,\ldots,x_k$ 
appearing to the patterns and to the variables $y_1,\ldots,y_m$ respectively. 
Since $x_i$ appears in an input position in $P$, it can appear only in 
an output position in the forms $!(F_1),\ldots !(F_k)$ and thus any variable of the term $\sigma_{\mathit{in}}(x_i)$ will appear in an 
output position in $!(F_i\sigma_{\mathit{in}})$. Similarly, since $y_i$ appears in an input position in the form 
$!(s(u_1,\ldots,u_n)\langle y_1,\ldots,y_m\rangle)$, it can only appear in an output position in $!(F)$ for the others 
forms $F$ of $\Gamma$. Consequently any variable of the term $\sigma_{\mathit{out}}(y_i)$ will appear in an 
output position in $!(F\sigma_{\mathit{out}})$ for any equation $X=F$ in $\Gamma$ with $X\neq X_0$. 
It follows that the application of a production cannot produce new occurrences of a variable in an input position 
and thus there cannot exist two occurrences $x^?$ of a same variable $x$ in $\Gamma'$. 
\PropBox{prop:firingruleAG}
\end{proof}

Thus applying an enabled production defines 
a binary relation on configurations.

\begin{definition}\label{def:accessibility}
A configuration $\Gamma'$ is \textbf{directly accessible} from $\Gamma$, denoted by $\netfir{\Gamma}{\,}{\Gamma'}$, whenever
$\netfir{\Gamma}{P/X}{\Gamma'}$ for some production $P$ enabled in node $X$ of configuration $\Gamma$. 
Furthermore, a configuration $\Gamma'$ is \textbf{accessible} from configuration $\Gamma$ when $\netfir{\Gamma}{*}{\Gamma'}$ 
where $\netfir{}{*}{}$ is the reflexive and transitive closure of relation $\netfir{}{\,}{}$.
\DefBox{def:accessibility}
\end{definition}

Recall that a substitution $\sigma$ unifies a set of equations $E$ if $t\sigma=t'\sigma$ for every equations $t=t'$ in $E$. 
A substitution $\sigma$ is more general than a substitution $\sigma'$ if $\sigma'=\sigma\sigma''$ for some substitution $\sigma''$.
If a system of equations has a some unifier, then it has (up to an bijective renaming of the variables in $\sigma$) 
a {\em most general unifier}. In particular a set of equations of the form $\setof{x_i=t_i}{1\leq i\leq n}$ has a unifier if 
and only if it is acyclic. In this case, the corresponding solved form is its most general unifier. 

\begin{proposition}
\label{prop:correctness}
If the left-hand side $F=s(p_1,\ldots,p_n)\langle u_1,\ldots, u_m\rangle$ of a production $P$ matches with the 
definition $X=s(d_1,\ldots,d_n)\langle y_1,\ldots,y_m\rangle$ of an open node $X$ then the substitution 
$\sigma=\mathbf{match}(F,X)$ is the most general unifier  of the set of equations 
$\setof{p_i=d_i}{1\leq i\leq n}\cup\setof{y_j=u_j}{1\leq j \leq m}$.
\end{proposition}

In order to prove Proposition~\ref{prop:correctness} we first recall some fact about unification.

\begin{recall}[on Unification]\label{recall:unification}
We consider sets  $E=E_?\uplus E_=$ containing equations of two kinds. An equation in $E_?$, denoted as $t\stackrel{?}{=}u$,
represents a \textit{unification goal} whose solution is a substitution $\sigma$ such that $t\sigma=u\sigma$, i.e.,\ substitution $\sigma$ unifies terms $t$ and $u$.
$E_=$ contains only equations of the form $x=t$ where variable $x$ occurs only there, i.e., we do not have two equations with the same variable in their 
left-hand side and such a variable cannot either occur in any right-hand side of an equation in $E_=$. A \textit{solution} to $E$ is any substitution $\sigma$ 
whose domain is the set of variables occurring in the right-hand sides of equations in $E_=$ such that the compound substitution made of $\sigma$ and 
the set of equations $\setof{x=t\sigma}{x=t\in E_=}$ unifies terms $t$ and $u$ for any equation $t\stackrel{?}{=}u$ in $E_?$. 
Two systems of equations are said to be \textit{equivalent} when they have the same solutions. 
A \textit{unification problem} is a set of such equations with $E_==\emptyset$, i.e.,\ it is a set of unification goals. 
On the contrary $E$ is said to be in \textit{solved form} if $E_?=\emptyset$, thus $E$ defines a substitution which, by definition, is the most general solution to $E$.  
Solving a unification problem $E$ consists in finding an equivalent system of equations $E'$ in solved form. 
In that case $E'$ is a \textit{most general unifier} for $E$. 

Martelli and Montanari Unification algorithm \cite{MartelliM82} proceeds as follows.
We pick up non deterministically one equation in $E_?$ and depending on its shape apply the corresponding transformation:
\begin{enumerate}
 \item $c(t_1,\ldots,t_n)\stackrel{?}{=}c(u_1,\ldots,u_n)$: replace it by equations $t_1\stackrel{?}{=}u_1,\ldots, t_1\stackrel{?}{=}u_1$.
 \item $c(t_1,\ldots,t_n)\stackrel{?}{=}c'(u_1,\ldots,u_m)$ with $c\neq c'$: halt with failure.
 \item $x\stackrel{?}{=}x$: delete this equation.
 \item $t\stackrel{?}{=}x$ where $t$ is not a variable: replace this equation by $x\stackrel{?}{=}t$.
 \item $x\stackrel{?}{=}t$ where $x\not\in \mathrm{var}(t)$: replace this equation by $x=t$ and substitute $x$ by $t$ in all other equations of $E$. 
 \item $x\stackrel{?}{=}t$ where $x\in\mathrm{var}(t)$ and $x\neq t$: halt with failure.
\end{enumerate}
The condition in (5) is the occur check. Thus the computation fails either if the two terms of an equation cannot be unified because their main constructors 
are different or because a potential solution of an equation is necessarily an infinite tree due to a recursive statement detected by the occur check. 
System $E'$  obtained from $E$ by applying one of these rules, denoted as $E\Rightarrow E'$, is clearly equivalent to $E$. 
We iterate this transformation as long as we do not encounter a failure and some equation remains in $E_?$. 
It can be proved that all these computations terminate and either the original unification problem $E$ has a solution (a unifier) and every computation 
terminates (and henceforth produces a solved set equivalent to $E$ describing a most general unifier of $E$) or $E$ has no unifier and every computation fails. 
We let
\[
\sigma = \mathbf{mgu}(\setof{t_i=u_i}{1\leq i\leq n})\;\;\mathbf{iff}\;\;
\setof{t_i\stackrel{?}{=}u_i}{1\leq i\leq n}\Rightarrow^* \sigma
\]
\RecallBox{recall:unification}
\end{recall}

Note that (5) and (6) are the only rules that can be applied to solve a unification problem of the form $\setof{y_i\stackrel{?}{=}u_i}{1\leq i\leq n}$, where the $y_i$ are 
distinct variables. The most general unifier exists  when the occur check always holds, i.e.,\ rule (5) always applies. The computation  amounts to 
iteratively replacing an occurrence of a variable $y_i$ in one of the right-hand side term $u_j$  by its definition $u_i$ until no variable $y_i$ occurs in some $u_j$. 
This process terminates precisely when the relation $y_i\succ y_j\Leftrightarrow y_j\in u_i$ is acyclic. When this condition is met we say that  
the set of equations $\setof{y_i=u_i}{1\leq i\leq n}$ \textit{is acyclic} and we say that it \textit{defines}  the substitution 
$\sigma=\mathbf{mgu}(\setof{y_i=u_i}{1\leq i\leq n})$. 

\begin{proof}[Proof of Proposition~\ref{prop:correctness}]\mbox{}\\
If a production $P$ of left-hand side $s(p_1,\ldots,p_n)\langle u_1,\ldots u_m\rangle$ is triggered in node $X_0$ defined by 
$X_0=s(d_1,\ldots,d_n)\langle y_1,\ldots,y_m\rangle$ then by Definition~\ref{def:firingRuleAG}
\[
 \setof{p_i\stackrel{?}{=}d_i}{1\leq i\leq n}\cup\setof{y_j\stackrel{?}{=}u_j}{1\leq j \leq m}\Rightarrow^* \sigma_{\mathit{in}} \cup
 \setof{y_j\stackrel{?}{=}u_j\sigma_{\mathit{in}}}{1\leq j\leq m}
\]
using only the rules (1) and (5). Now 
\[
 \sigma_{\mathit{in}} \cup \setof{y_j\stackrel{?}{=}u_j\sigma_{\mathit{in}}}{1\leq j\leq m} \Rightarrow^* 
 \sigma_{\mathit{in}} \cup \mathbf{mgu}\setof{y_j=u_j\sigma_{\mathit{in}}}{1\leq j\leq m}
\]
by applying iteratively rule~(5) if the set of equations $\setof{y_j=u_j\sigma_{\mathit{in}}}{1\leq j\leq m}$ 
satisfies the occur check. Then $\sigma_{\mathit{in}}+\sigma_{\mathit{out}}\Rightarrow^* \sigma$ again by using rule~(5).
\PropBox{prop:correctness}
\end{proof}

Note that the converse does not hold. Namely, one shall not deduce from Proposition~\ref{prop:correctness} that the 
relation $\netfir{\Gamma}{P/X_0}{\Gamma'}$ is defined  whenever the left-hand side $\mathrm{lhs}(P)$ of $P$  can 
be unified with the definition $\mathrm{def}(X_0,\Gamma)$ of $X_0$ in $\Gamma$ with
\[
  \begin{array}{rcl}
   \Gamma' & = & \set{X_0=P(X_1,\ldots,X_k)}\;\; \mbox{where}\;\; X_1,\ldots,X_k\quad\mbox{are nodes added to $\Gamma'$}\\
           & \cup & \set{X_1=F_1\sigma, \ldots, X_k=F_k\sigma} \\
           & \cup & \setof{X=F\sigma}{(X=F)\in\Gamma\;\wedge\; X\neq X_0}
  \end{array}
\]
where $\sigma=\mathbf{mgu}(\mathrm{lhs}(P),\mathrm{def}(X_0,\Gamma))$ is the corresponding most general unifier. 
Indeed, when unifying $s(d_1,\ldots,d_n,y_1,\ldots,y_m)$ with $s(p_1,\ldots, p_n,u_1,\ldots, u_m)$ one may generate an 
equation of the form $x=t$ where $x$ is a variable in an inherited data $d_i$ and $t$ is an instance of a corresponding subterm in the associated pattern $p_i$. 
This would correspond to a situation where information is sent to the context of a node through one of its inherited attribute!
Otherwise stated some parts of the pattern $p_i$ are actually used to filtered out the incoming data value $d_i$ while some other parts of the same 
pattern are used to transfert synthesized information to the context.  

As already mentioned, an artifact is refined by applying a production at one of its open node. However we also need means to
initiate cases. To this extent, we define interfaces for GAGs, that describe how services can initialize new artifacts.

\begin{definition}
\label{def:interface}
The \textbf{interface} of a guarded attribute grammar is given by a subset $\mathcal{I}$ of forms 
$F=s(t_1,\ldots,t_n)\langle x_1,\ldots,x_m\rangle$, called its \textbf{services} 
where the synthesized positions are (distinct) variables $ x_1,\ldots,x_m$. 
The invocation of the service produces a new artifact reduced to a single open node defined by $F$, 
it is associated with \textbf{initial configuration} $$\Gamma_0=\set{X_0=s(t_1,\ldots,t_n)\langle x_1,\ldots,x_m\rangle}$$ 
An \textbf{accessible configuration} of a guarded attribute grammar is a configuration accessible from 
one of its initial configurations.  
\DefBox{def:interface}
\end{definition}

\begin{example}\label{exple:flattening3}
The attribute grammar for the flattening of a binary tree (Example~\ref{exple:flattening}) can be presented as a guarded  
attribute grammar with the following productions:
\[
\begin{array}{l@{:\quad}r@{\leftarrow}l}
  \mathrm{Fork}    & \mathit{bin}(x)\langle y\rangle &\mathit{bin}(z)\langle y\rangle\;\mathit{bin}(x)\langle z\rangle\\
  \mathrm{Leaf}_a  & \mathit{bin}(x)\langle \mathrm{Cons}_a(x)\rangle &
\end{array}
\]
together with service $\mathrm{Init}\langle x\rangle=\mathit{bin}(Nil)\langle x\rangle$ whose invocation creates a new 
binary tree and expects for the list of its leaves. 
In comparison with Example~\ref{exple:flattening} we avoid the construction of the artificial node \textrm{Root} 
whose sole purpose was to initialize the inherited attribute of the tree at its root.
\ExpleBox{exple:flattening3}
\end{example}

%% file: examples.tex
In this section we illustrate the behaviour of guarded attribute grammars with three examples.

Example~\ref{exple:flattening2} describes an execution of the attribute grammar of Example~\ref{exple:flattening}.  
The specification in Example~\ref{exple:flattening} is actually an ordinary attribute grammar because   
the inherited attributes in the left-hand sides of productions are plain variables. 
This example shows how data are lazily produced and send in push mode through attributes. 
It also illustrates the role of the data links and their dynamic evolutions.

Example~\ref{exple:coprocess} illustrates the role of the guards by describing two processes acting as coroutines.
The first process sends forth a list of values to the second process and it waits for an acknowledgement 
for each message before sending the next one.

Example~\ref{exple:OC} justifies the role of the occur check.

\input{flattening2}
The above example shows that data links  are used to transmit data  in push mode from 
a source vertex $v$ (the input occurrence $x^?$ of a variable $x$) 
to some target vertex $v'$ (an output occurrence $x^!$ of the same variable). 
These links $(x^!,x^?)$ are {\em transient} in the sense that they disappear as soon as 
variable $x$ gets defined by the substitution $\sigma_{\mathit{out}}$ induced by the 
application of a production in some open node of the current configuration. 
If $\sigma_{\mathit{out}}(x)$  is a term $t$, not reduced to a variable, with variables $x_1,\ldots,x_k$ then 
vertex $v'$  is refined by the term $t[x_i^!/x_i]$ and new vertices $v'_i$ 
~---associated with these new occurrences of $x_i$ in an output position---~ are created. 
The original data link $(x^?,x^!)$ is replaced by all the corresponding instances of $(x_i^?,x_i^!)$.  
Consequently, a target is replaced by new targets  which are the recipients for the subsequent pieces of information  
(maybe none because no new links are created when $t$ contains no variable).
If the term $t$ is a variable $y$ then the link $(x^?,x^!)$ is replaced by the link $(y^?,y^!)$ with the same target 
and whose source, the (unique) occurrence $x^?$ of variable $x$, is replaced by the (unique) occurrence $y^?$ of variable $y$. 
Therefore the direction of the flow of information is in both cases preserved: Channels can be viewed as ``generalized streams'' 
(that can fork or vanish) through which information is pushed incrementally. 

\input{coprocess}

We  say that a production $P$ is {\bf triggered} in node $X$ if substitution $\sigma_{\mathit{in}}$ is defined,
i.e.,\ the patterns $p_i$ match the data $d_i$. 
As shown by the following example one can usually suspect an error in the specification 
when a triggered transition is not enabled due to the fact that 
the system of equations $\setof{y_j=u_j\sigma_{\mathit{in}}}{1\leq j\leq m}$ is cyclic. 

\begin{example}\label{exple:OC}
Let us consider the guarded attribute grammar given by the following productions:
\[
\begin{array}{c@{\; : \quad}l}
P & s_0(\,)\langle\,\rangle \leftarrow s_1(a(x))\langle x\rangle\quad s_2(x)\langle\,\rangle\\
Q & s_1(y)\langle a(y)\rangle\quad \leftarrow \\
R & s_2(a(z))\langle \,\rangle\quad \leftarrow 
\end{array}
\]
Applying production $P_0$ in node $X_0$ of configuration $\Gamma_0=\set{X_0=s_0(\,)\langle\,\rangle}$ leads to configuration
\[
 \Gamma_1=\set{X_0=P(X_1,X_2);\;X_1=s_1(a(x))\langle x\rangle;\; X_2=s_2(x)\langle\,\rangle}
\]
Production $Q$ is triggered in node $X_1$ with $\sigma_{\mathit{in}}=\set{y=a(x)}$ but the occur check fails 
because variable $x$ occurs in $a(y)\sigma_{\mathit{in}}=a(a(x))$. 
Alternatively, we could drop the occur check and instead  adapt the fixed point semantics for attribute evaluation defined in~\cite{ChiricaM79,Mayoh81} 
in order to cope with infinite data structures. More precisely we could let 
$\sigma_{\mathit{out}}$ be defined as the least solution of the system of equations $\setof{y_i=u_j\sigma_{\mathit{in}}}{1\leq j\leq m}$ 
~---assuming these equations are guarded, i.e.,\ that there is no cycle of copy rules in the link graph of any accessible configuration---. 
In that case the infinite tree $a^{\omega}$ is substituted to variable $x$ and the unique maximal computation associated with the grammar 
is given by the infinite tree $P(Q,R^{\omega})$. 
In Definition~\ref{def:firingRuleAG} we have chosen to restrict ourself to finite data structures which seems a reasonable assumption 
in view of the nature of systems we want to model. The occur check is used to avoid recursive definitions of attribute values.
The given example, whose most natural interpretation is given by fixed point computation, should in that respect be considered as ill-formed.
And indeed this guarded attribute grammar is not \textit{sound} (a notion presented in Section~\ref{sec:soundness}) because the configuration 
$\Gamma$ is not closed (it still contains open nodes), hence it represents a case that is not terminated. However it is a terminal configuration 
since it enables no production. 
\ExpleBox{exple:OC}
\end{example}

%% file: flattening2.tex
\begin{example}[Example~\ref{exple:flattening} continued]\label{exple:flattening2}
Let us consider the attribute grammar of Example~\ref{exple:flattening} and the initial configuration 
$\Gamma_0=\set{X_0=\mathit{root}()\langle x\rangle,Y_0=\mathit{toor}(x)\langle \rangle}$ shown \vspace{-5mm}next\\
\begin{center}
\begin{tikzpicture}[>=stealth',shorten >=1pt, auto,on grid,>=stealth',auto,bend angle=10,
                    every place/.style={minimum size=4.5mm,thick,draw=black!60,inner sep=.2pt},
                    every transition/.style={minimum width=5mm,minimum height=5mm,inner sep=2pt,draw=black!60,thick,fill=white},
                    state/.style={circle,minimum size=6mm,inner sep=.5pt,draw=black!25,thick,fill=black!10}]

\begin{scope}[scale=.6]
\node[state,label=left:{\scriptsize $X_0::\mathit{root}$}]  (x) at (0,2)   {$?$};
\node[state,label=right:{\scriptsize $Y_0::\mathit{toor}$}]  (x1) at (4,2)   {$?$};
%\node[place,draw=green!45,thick,fill=green!10]  (inhx) at (-.8,1)   {\tiny ROOT};
%\node[place,draw=green!45,thick,fill=green!10]  (u) at (-.8,0)   {\tiny $U_1$};
\node[place,draw=green!45,thick,fill=green!10]  (y) at (3.2,1)   {\tiny $x$};
\node[place,draw=red!45,thick,fill=red!10,label=80:{\tiny $x$}]  (syn) at (.8,1) {\scriptsize $?$};
\path[-] (x)  edge (syn);
%              edge (inhx);
%\draw[-,color=black!40,thick] (inhx)--(u);
\path[-] (x1)  edge (y);
\path[->,thick,color=black!40] (syn) edge[bend left] (y);
\end{scope}
\end{tikzpicture}
\end{center}
The annotated version $!\Gamma_0=\setof{!F}{F\in\Gamma_0}$ of configuration $\Gamma_0$ is 
$$!\Gamma_0=\set{X_0=\mathit{root}()\langle x^?\rangle,Y_0=\mathit{toor}(x^!)\langle \rangle}$$ 
The data link from $x^?$ to $x^!$ says that the list of the leaves of the tree  
~---that will stem from node $X_0$---~   
to be synthesized at node $X_0$ should be forwarded to the inherited attribute of $Y_0$.

This tree is not defined in the intial configuration $\Gamma_0$. 
One can start developping it by applying production 
$\mathrm{Root} : \mathit{root}()\langle u\rangle \leftarrow\mathit{bin}(Nil)\langle u\rangle$ 
at node $X_0::\mathit{root}$. Actually  the left-hand side $\mathit{root}()\langle u\rangle$ 
of production $\mathrm{Root}$ matches 
with the definition $\mathit{root}()\langle x\rangle$ of $X_0$ 
with $\sigma_{\mathit{in}}=\emptyset$ and $\sigma_{\mathit{out}}=\set{x=u}$. 
Hence $\netfir{\Gamma_0}{\mathrm{Root}/X_0}{\Gamma_1}$ where the  
annotated configuration $!\Gamma_1$ is given in Figure~\ref{fig:Gamma1}. 
\begin{figure}[htb]
\parbox[c]{0.40\linewidth}{
\begin{tikzpicture}[>=stealth',shorten >=1pt, auto,on grid,>=stealth',auto,bend angle=10,
                    every place/.style={minimum size=4.5mm,thick,draw=black!60,inner sep=.2pt},
                    every transition/.style={minimum width=5mm,minimum height=5mm,inner sep=2pt,draw=black!60,thick,fill=white},
                    state/.style={circle,minimum size=6mm,inner sep=.5pt,draw=black!25,thick,fill=black!10}]

\begin{scope}[scale=.6]
\node[state,label=left:{\scriptsize $X_0::root$}]  (x) at (0,1.5)   {\tiny $\mathrm{Root}$};
\node[state,label=left:{\scriptsize $X_1::bin$}]  (x1) at (0,0)  {\scriptsize $?$};
\path[-,color=black!40,very thick] (x)  edge (x1);
                                       
\node[state,label=right:{\scriptsize $Y_0::\mathit{toor}$}]  (yy) at (4,1.5)   {\tiny $?$};
\node[place,draw=green!45,thick,fill=green!10]  (y) at (3.2,.5)   {\tiny $u$};
\path[-] (yy)  edge (y);
               
\node[place,draw=green!45,thick,fill=green!10]  (inh1) at (-.75,-.75)  {\tiny $\mathrm{Nil}$};
%\node[place,draw=green!45,thick,fill=green!10]  (u) at (-1.5,-.75)   {\tiny $U_1$};
\node[place,draw=red!45,thick,fill=red!10,label=right:{\tiny $u$}]  (syn1) at (.75,-.75) {\tiny $?$};
\path[-] (x1)  edge (inh1)
               edge (syn1);

\path[->,thick,color=black!40] (syn1) edge[bend left] (y);
\end{scope}
\end{tikzpicture}
}\hfill
\parbox[c]{0.40\linewidth}{
\[
\begin{array}{r@{\;=\;}l}
X_0 & \mathrm{Root}(X_1)\\
X_1 & \mathit{bin}(\mathrm{Nil})\langle u^?\rangle \\
Y_0 & \mathit{toor}(u^!)\langle\,\rangle
\end{array}
\]
}
\caption{Configuration $\Gamma_1$}
\label{fig:Gamma1}
\end{figure}
Note that substitution $\sigma_{\mathrm{out}}=\set{x=u}$ replaces the data link $(x^?,x^!)$ by a new link 
$(u^?,u^!)$ with the same target and whose source has been moved from the synthesized attribute of 
$X_0$ to the synthesized attribute of $X_1$\vspace{-3mm}:\\
\begin{center}
\begin{tikzpicture}[>=stealth',shorten >=1pt, auto,on grid,>=stealth',auto,bend angle=30,
                    every place/.style={minimum size=4.5mm,thick,draw=black!60,inner sep=.2pt}]

\begin{scope}[scale=.8]
\node[place,draw=black!35,thick,fill=black!5,label=150:{\small $u$}]  (usource) at (-3,0)  {\small $?$};
\node[place,draw=black!35,thick,fill=black!5,label=150:{\small $x$}]  (source) at (-1,0)  {\small $u$};
\node[place,draw=black!35,thick,fill=black!5]  (target) at (1,0)  {\small $x$};
\path[->,thick,color=black!40] (usource) edge (source); 
\path[->,thick,color=black!40] (source) edge (target); 
\node at (3,0) {$\netfir{}{\sigma_{\mathrm{out}}}{}$};
\node[place,draw=black!35,thick,fill=black!5,label=150:{\small $u$}]  (so) at (5,0)  {\small $?$};
\node[place,draw=black!35,thick,fill=black!5]  (ta) at (7,0)  {\small $u$};
\path[->,thick,color=black!40] (so) edge (ta); 
\end{scope}
\end{tikzpicture}
\end{center} 

The tree may be further refined by applying production 
$\mathrm{Fork} : \mathit{bin}(x)\langle y\rangle \leftarrow \mathit{bin}(z)\langle y\rangle \;\mathit{bin}(x)\langle z\rangle$  at 
node $X_1::\mathit{bin}$ since its left-hand side $\mathit{bin}(x)\langle y\rangle$ matches with 
the definition $\mathit{bin}(\mathrm{Nil})\langle u\rangle$ of $X_1$ with 
$\sigma_{\mathit{in}}=\set{x=Nil}$ and $\sigma_{\mathit{out}}=\set{u=y}$. 
Hence $\netfir{\Gamma_1}{\mathrm{Fork}/X_1}{\Gamma_2}$ where  $!\Gamma_2$ is given in Figure~\ref{fig:Gamma2}.
\begin{figure}[htb]
\parbox[c]{0.40\linewidth}{
\begin{tikzpicture}[>=stealth',shorten >=1pt, auto,on grid,>=stealth',auto,bend angle=80,
                    every place/.style={minimum size=4.5mm,thick,draw=black!60,inner sep=.2pt},
                    every transition/.style={minimum width=5mm,minimum height=5mm,inner sep=2pt,draw=black!60,thick,fill=white},
                    state/.style={circle,minimum size=6mm,inner sep=.5pt,draw=black!25,thick,fill=black!10}]

\begin{scope}[scale=.6]
\node[state,label=left:{\scriptsize $X_0::\mathit{root}$}]  (x0) at (0,3.5)   {\tiny $\mathrm{Root}$};
\node[state]  (x) at (0,2)   {\tiny $\mathrm{Fork}$};
\path[-,color=black!40,very thick] (x0)  edge (x);

\node[state,label=right:{\scriptsize $Y_0::\mathit{toor}$}]  (yy) at (4,3.5)   {\scriptsize $?$};
\node[place,draw=green!45,thick,fill=green!10]  (y) at (3.2,2.5)   {\tiny $y$};
\path[-] (yy)  edge (y);

\node[state]  (x1) at (-1.5,0)  {\scriptsize $?$};
\node[state]  (x2) at (1.5,0) {\scriptsize $?$};
\path[-,color=black!40,very thick] (x)  edge (x1)
                                        edge (x2);

\node[place,draw=green!45,thick,fill=green!10]  (inh1) at (-2.25,-.75)  {\tiny $z$};
%\node[place,draw=green!45,thick,fill=green!10]  (u2) at (-3,-.75)  {\tiny $U_2$};
\node[place,draw=red!45,thick,fill=red!10,label=80:{\tiny $y$}]  (syn1) at (-.75,-.75) {\tiny $?$};
\path[-] (x1)  edge (inh1)
               edge (syn1);

\node[place,draw=green!45,thick,fill=green!10]  (inh2) at (.75,-.75)  {\tiny $\mathrm{Nil}$};
%\node[place,draw=green!45,thick,fill=green!10]  (u3) at (0,-.75)  {\tiny $U_3$};
\node[place,draw=red!45,thick,fill=red!10,label=70:{\tiny $z$}]  (syn2) at (2.25,-.75) {\scriptsize $?$};
\path[-] (x2)  edge (inh2)
               edge (syn2);
               
\path[->,thick,color=black!40] (syn1) edge (y);
\path[->,thick,color=black!40] (syn2.70) edge[bend right] (inh1.110);
\end{scope}
\end{tikzpicture}
}\hfill
\parbox[c]{0.40\linewidth}{
\[
\begin{array}{r@{\;=\;}l}
X_0 & \mathrm{Root}(X_1)\\
X_1 & \mathrm{Fork}(X_{11},X_{12})\\
X_{11} & \mathit{bin}(z^!)\langle y^?\rangle\\
X_{12} & \mathit{bin}(Nil)\langle z^?\rangle\\
Y_0 & \mathit{toor}(y^!)\langle\,\rangle
\end{array}
\]
}
\caption{Configuration $\Gamma_2$}
\label{fig:Gamma2}
\end{figure}

Production $\mathrm{Leaf}_c : \mathit{bin}(x)\langle\mathrm{Cons}_c(x)\rangle\leftarrow $ applies at node 
$X_{12}$ since its left-hand side $\mathit{bin}(x)\langle\mathrm{Cons}_c(x)\rangle$ matches with the definition 
$\mathit{bin}(Nil)\langle z\rangle$ of $X_{12}$ 
with $\sigma_{\mathit{in}}=\set{x=Nil}$ and $\sigma_{\mathit{out}}=\set{z=\mathrm{Cons}_c(\mathrm{Nil})}$. 
Hence $\netfir{\Gamma_2}{\mathrm{Leaf}_c/X_{12}}{\Gamma_3}$ where the annotated configuration $!\Gamma_3$ is given in Figure~\ref{fig:Gamma3}.
\begin{figure}[htb]
\parbox[c]{0.40\linewidth}{
\begin{tikzpicture}[>=stealth',shorten >=1pt, auto,on grid,>=stealth',auto,bend angle=80,
                    every place/.style={minimum size=4.5mm,thick,draw=black!60,inner sep=.2pt},
                    every transition/.style={minimum width=5mm,minimum height=5mm,inner sep=2pt,draw=black!60,thick,fill=white},
                    state/.style={circle,minimum size=6mm,inner sep=.5pt,draw=black!25,thick,fill=black!10}]

\begin{scope}[scale=.6]
\node[state,label=left:{\scriptsize $X_0::\mathit{root}$}]  (x0) at (0,3.5)   {\tiny $\mathrm{Root}$};
\node[state]  (x) at (0,2)   {\tiny $\mathrm{Fork}$};
\path[-,color=black!40,very thick] (x0)  edge (x);

\node[state,label=right:{\scriptsize $Y_0::\mathit{toor}$}]  (yy) at (4,3.5)   {\scriptsize $?$};
\node[place,draw=green!45,thick,fill=green!10]  (y) at (3.2,2.5)   {\tiny $y$};
\path[-] (yy)  edge (y);

\node[state]  (x1) at (-1.5,0)  {\tiny $?$};
\node[state]  (x2) at (1.5,0) {\tiny $\mathrm{Leaf}_c$};
\path[-,color=black!40,very thick] (x)  edge (x1)
                                        edge (x2);

\node[place,draw=green!45,thick,fill=green!10]  (inh1) at (-2.50,-.75)  {\tiny $\mathrm{Cons}_c$};
\node[place,draw=green!45,thick,fill=green!10]  (inh11) at (-2.50,-1.75)  {\tiny $\mathrm{Nil}$};
\path[-,color=green!50,thick] (inh1)  edge (inh11); 
%\node[place,draw=green!45,thick,fill=green!10]  (u2) at (-3.2,-.75)  {\tiny $U_2$};
\node[place,draw=red!45,thick,fill=red!10,label=right:{\tiny $y$}]  (syn1) at (-.75,-.75) {\tiny $?$};
\path[-] (x1)  edge (inh1)
               edge (syn1);

\path[->,thick,color=black!40] (syn1) edge (y);

\end{scope}
\end{tikzpicture}
}\hfill
\parbox[c]{0.40\linewidth}{
\[
\begin{array}{r@{\;=\;}l}
X_0 & \mathrm{Root}(X_1)\\
X_1 & \mathrm{Fork}(X_{11},X_{12})\\
X_{11} & \mathit{bin}(\mathrm{Cons}_c(\mathrm{Nil}))\langle y^?\rangle\\
X_{12} & \mathrm{Leaf}_c\\
Y_0 & \mathit{toor}(y^!)\langle\,\rangle
\end{array}
\]
}
\caption{Configuration $\Gamma_3$}
\label{fig:Gamma3}
\end{figure}
As a result of substitution $\sigma_{\mathit{out}}=\set{z=\mathrm{Cons}_c(\mathrm{Nil})}$ 
the value $\mathrm{Cons}_c(\mathrm{Nil})$ is transmitted through the link 
$(z^?,z^!)$ and this link disappears. 

Production $\mathrm{Fork} : \mathit{bin}(x)\langle u\rangle \leftarrow \mathit{bin}(z)\langle u\rangle \;\mathit{bin}(x)\langle z\rangle$ may apply at 
node $X_{11}$ since its left-hand side $\mathit{bin}(x)\langle u\rangle$ matches with 
the definition $\mathit{bin}(\mathrm{Cons}_c(\mathrm{Nil}))\langle y\rangle$ of $X_{11}$ with 
$\sigma_{\mathit{in}}=\set{x=\mathrm{Cons}_c(\mathrm{Nil}))}$ and $\sigma_{\mathit{out}}=\set{y=u}$. 
Hence $\netfir{\Gamma_3}{\mathrm{Fork}/X_1}{\Gamma_4}$ with configuration $?\Gamma_4$ given in Figure~\ref{fig:Gamma4}.
\begin{figure}[htb]
\parbox[c]{0.40\linewidth}{
\begin{tikzpicture}[>=stealth',shorten >=1pt, auto,on grid,>=stealth',auto,bend angle=80,
                    every place/.style={minimum size=4.5mm,thick,draw=black!60,inner sep=.2pt},
                    every transition/.style={minimum width=5mm,minimum height=5mm,inner sep=2pt,draw=black!60,thick,fill=white},
                    state/.style={circle,minimum size=6mm,inner sep=.5pt,draw=black!25,thick,fill=black!10}]

\begin{scope}[scale=.6]
\node[state,label=left:{\scriptsize $X_0::\mathit{root}$}]  (x0) at (0,3.5)   {\tiny $\mathrm{Root}$};
\node[state]  (x) at (0,2)   {\tiny $\mathrm{Fork}$};
\path[-,color=black!40,very thick] (x0)  edge (x);
 
\node[state,label=right:{\scriptsize $Y_0::\mathit{toor}$}]  (yy) at (4,3.5)   {\scriptsize $?$};
\node[place,draw=green!45,thick,fill=green!10]  (y) at (3.2,2.5)   {\tiny $u$};
\path[-] (yy)  edge (y);   

%\node[place,draw=black!25,thick,fill=black!10,label=above:{\scriptsize $y$}]  (y) at (3,3.5) {\scriptsize $?$};

\node[state,draw=black!25,thick,fill=black!10,inner sep=.2pt]  (x2) at (1.5,.3)   {\tiny $\mathrm{Leaf}_c$};
 
\node[state]  (x1) at (-1.5,.3)   {\tiny $\mathrm{Fork}$};
\node[state]  (x11) at (-3,-1.7)  {\tiny $?$};
\node[state]  (x12) at (0,-1.7) {\tiny $?$};
\path[-,color=black!40,very thick] (x1)  edge (x11)
                                         edge (x12);

\path[-,color=black!40,very thick] (x)  edge (x1)
                                        edge (x2);

\node[place,draw=green!45,thick,fill=green!10]  (inh11) at (-3.75,-2.45)  {\tiny $z$};
\node[place,draw=red!45,thick,fill=red!10,label=below:{\tiny $u$}]  (syn11) at (-2.25,-2.45) {\tiny $?$};
\path[-] (x11)  edge (inh11)
                edge (syn11);

\node[place,draw=green!45,thick,fill=green!10,inner sep=.2pt]  (inh12) at (-.9,-2.45)   {\tiny $\mathrm{Cons}_c$};              
\node[place,draw=green!45,thick,fill=green!10]  (nil) at (-.9,-3.55)  {\tiny $\mathrm{Nil}$};

\path[-,color=black!40,thick] (inh12) edge (nil);

\node[place,draw=red!45,thick,fill=red!10,label=below:{\tiny $z$}]  (syn12) at (.75,-2.45) {\tiny $?$};
\path[-] (x12)  edge (inh12)
               edge (syn12);
               
\path[->,color=black!40,thick] (syn12.80) edge[bend right] (inh11.100);
 
\path[->,thick,color=black!40] (syn11) edge[bend angle=10,bend left] (y);

\end{scope}
\end{tikzpicture}
}\hfill
\parbox[c]{0.40\linewidth}{
\[
\begin{array}{r@{\;=\;}l}
X_0 & \mathrm{Root}(X_1)\\
X_1 & \mathrm{Fork}(X_{11},X_{12})\\
X_{11} & \mathrm{Fork}(X_{111},X_{112})\\
X_{111} & \mathit{bin}(z!)\langle u?\rangle\\
X_{112} & \mathit{bin}(\mathrm{Cons}_c(\mathrm{Nil}))\langle z?\rangle\\
X_{12} & \mathrm{Leaf}_c\\
Y_0 & \mathit{toor}(u!)\langle\,\rangle
\end{array}
\]
}
\caption{Configuration $\Gamma_4$}
\label{fig:Gamma4}
\end{figure}

Production $\mathrm{Leaf}_a : \mathit{bin}(x)\langle\mathrm{Cons}_a(x)\rangle\leftarrow $ applies at node 
$X_{111}$ since its left-hand side $\mathit{bin}(x)\langle\mathrm{Cons}_a(x)\rangle$ matches with the definition 
$\mathit{bin}(z)\langle u\rangle$ of $X_{111}$ 
with $\sigma_{\mathit{in}}=\set{x=z}$ and $\sigma_{\mathit{out}}=\set{u=\mathrm{Cons}_a(z)}$. 
Hence $\netfir{\Gamma_4}{\mathrm{Leaf}_a/X_{111}}{\Gamma_5}$ with configuration $!\Gamma_5$ given in Figure~\ref{fig:Gamma5}.
\begin{figure}[htb]
\parbox[c]{0.40\linewidth}{
\begin{tikzpicture}[>=stealth',shorten >=1pt, auto,on grid,>=stealth',auto,bend angle=20,
                    every place/.style={minimum size=4.5mm,thick,draw=black!60,inner sep=.2pt},
                    every transition/.style={minimum width=5mm,minimum height=5mm,inner sep=2pt,draw=black!60,thick,fill=white},
                    state/.style={circle,minimum size=6mm,inner sep=.5pt,draw=black!25,thick,fill=black!10}]

\begin{scope}[scale=.6]
\node[state,label=left:{\scriptsize $X_0::\mathit{root}$}]  (x0) at (0,3.5)   {\tiny $\mathrm{Root}$};
\node[state]  (x) at (0,2)   {\tiny $\mathrm{Fork}$};
\path[-,color=black!40,very thick] (x0)  edge (x);
                                       
\node[state,label=right:{\scriptsize $Y_0::\mathit{toor}$}]  (yy) at (4,3.5)   {\scriptsize $?$};
\node[place,draw=green!25,thick,fill=green!10,inner sep=.2pt]  (y) at (3.2,2.5)   {\tiny $\mathrm{Cons}_a$};
\node[place,draw=green!25,thick,fill=green!10]  (y1) at  (3.2,1)  {\tiny $z$};
\path[-,color=black!40,very thick] (y) edge (y1);
\path[-] (yy)  edge (y);

\node[state,draw=black!25,thick,fill=black!10,inner sep=.2pt]  (x2) at (1.5,.3)   {\tiny $\mathrm{Leaf}_c$};
\node[state]  (x1) at (-1.5,.3)   {\tiny $\mathrm{Fork}$}; 
\node[state,draw=black!25,thick,fill=black!10,inner sep=.2pt]  (x11) at (-3,-1.7)    {\tiny $\mathrm{Leaf}_a$};
\node[state,label=right:{\scriptsize $z$}]  (x12) at (0,-1.7)  {\tiny $?$};
                                         
\node[place,draw=red!45,thick,fill=red!10]  (syn12) at (1,-2.7) {\tiny $?$};

\node[place,draw=green!45,thick,fill=green!10,inner sep=.2pt]  (inh12) at (-1.2,-2.7)   {\tiny $\mathrm{Cons}_c$};              
\node[place,draw=green!45,thick,fill=green!10]  (nil) at (-1.2,-3.8)  {\tiny $\mathrm{Nil}$};

\path[-,color=black!40,thick] (inh12) edge (nil);

\path[-] (x12)  edge (inh12)
                edge (syn12);

\path[-,color=black!40,very thick] (x1)  edge (x11)
                                         edge (x12);

\path[-,color=black!40,very thick] (x)  edge (x1)
                                        edge (x2);
\path[->,thick,color=black!40] (syn12) edge (y1);

\end{scope}
\end{tikzpicture}
}\hfill
\parbox[c]{0.40\linewidth}{
\[
\begin{array}{r@{\;=\;}l}
X_0 & \mathrm{Root}(X_1)\\
X_1 & \mathrm{Fork}(X_{11},X_{12})\\
X_{11} & \mathrm{Fork}(X_{111},X_{112})\\
X_{111} & \mathrm{Leaf}_a\\
X_{112} & \mathit{bin}(\mathrm{Cons}_c(\mathrm{Nil}))\langle z^?\rangle\\
X_{12} & \mathrm{Leaf}_c\\
Y_0 & \mathit{toor}(\mathrm{Cons}_a(z^!))\langle\,\rangle
\end{array}
\]
}
\caption{Configuration $\Gamma_5$}
\label{fig:Gamma5}
\end{figure}
Using substitution $\sigma_{\mathit{out}}=\set{u=\mathrm{Cons}_a(z)}$ 
the data $\mathrm{Cons}_a(z)$ is transmitted through the link $(u^?,u^!)$ 
which, as a result, disappears. A new link $(z^?,z^!)$ is created so that the 
rest of the list, to be synthesized in node $X_{112}$ can later be forwarded 
to the inherited attribute of $Y_0$. 

Finally one can apply production $\mathrm{Leaf}_b : \mathit{bin}(x)\langle\mathrm{Cons}_a(x)\rangle\leftarrow $  at node 
$X_{112}$ since its left-hand side  matches with the definition 
$\mathit{bin}(\mathrm{Cons}_c(\mathrm{Nil}))\langle z\rangle$ of $X_{112}$ 
with $\sigma_{\mathit{in}}=\set{x=\mathrm{Cons}_c(\mathrm{Nil})}$ and $\sigma_{\mathit{out}}=\set{z=\mathrm{Cons}_b(\mathrm{Cons}_c(\mathrm{Nil}))}$. 
Therefore $\netfir{\Gamma_5}{\mathrm{Leaf}_b/X_{112}}{\Gamma_6}$ with configuration $!\Gamma_6$ given in Figure~\ref{fig:Gamma6}.
\begin{figure}[htb]
\parbox[c]{0.40\linewidth}{
\begin{tikzpicture}[>=stealth',shorten >=1pt, auto,on grid,>=stealth',auto,bend angle=20,
                    every place/.style={minimum size=4.5mm,thick,draw=black!60,inner sep=.2pt},
                    every transition/.style={minimum width=5mm,minimum height=5mm,inner sep=2pt,draw=black!60,thick,fill=white},
                    state/.style={circle,minimum size=6mm,inner sep=.5pt,draw=black!25,thick,fill=black!10}]

\begin{scope}[scale=.6]
\node[state,label=110:{\scriptsize $X_0::\mathit{root}$}]  (x0) at (0,3.5)   {\tiny $\mathrm{Root}$};
\node[state]  (x) at (0,2)   {\tiny $\mathrm{Fork}$};
\path[-,color=black!40,very thick] (x0)  edge (x);
                                       
\node[state,label=110:{\scriptsize $Y_0::\mathit{toor}$}]  (yy) at (4,3.5)   {\scriptsize $?$};
\node[place,draw=green!25,thick,fill=green!10,inner sep=.2pt]  (y) at (3.2,2.5)   {\tiny $\mathrm{Cons}_a$};
\node[place,draw=green!25,thick,fill=green!10]  (y1) at  (3.2,1)  {\tiny $\mathrm{Cons}_b$};
\node[place,draw=green!25,thick,fill=green!10]  (y11) at  (3.2,-.5)  {\tiny $\mathrm{Cons}_c$};
\node[place,draw=green!25,thick,fill=green!10]  (y111) at  (3.2,-2)  {\tiny $\mathrm{Nil}$};
\path[-,color=black!40,very thick] (y) edge (y1);
\path[-,color=black!40,very thick] (y1) edge (y11);
\path[-,color=black!40,very thick] (y11) edge (y111);
\path[-] (yy)  edge (y);

\node[state,draw=black!25,thick,fill=black!10,inner sep=.2pt]  (x2) at (1.5,.3)   {\tiny $\mathrm{Leaf}_c$};
\node[state]  (x1) at (-1.5,.3)   {\tiny $\mathrm{Fork}$}; 
\node[state,draw=black!25,thick,fill=black!10,inner sep=.2pt]  (x11) at (-3,-1.7)    {\tiny $\mathrm{Leaf}_a$};
\node[state]  (x12) at (0,-1.7)  {\tiny $\mathrm{Leaf}_b$};

\path[-,color=black!40,very thick] (x1)  edge (x11)
                                         edge (x12);

\path[-,color=black!40,very thick] (x)  edge (x1)
                                        edge (x2);

\end{scope}
\end{tikzpicture}
}\hfill
\parbox[c]{0.55\linewidth}{
\[
\begin{array}{r@{\;=\;}l}
X_0 & \mathrm{Root}(X_1)\\
X_1 & \mathrm{Fork}(X_{11},X_{12})\\
X_{11} & \mathrm{Fork}(X_{111},X_{112})\\
X_{111} & \mathrm{Leaf}_a\\
X_{112} & \mathrm{Leaf}_b\\
X_{12} & \mathrm{Leaf}_c\\
Y_0 & \mathit{toor}(\mathrm{Cons}_a(\mathrm{Cons}_b(\mathrm{Cons}_c(\mathrm{Nil}))))\langle\,\rangle
\end{array}
\]
}
\caption{Configuration $\Gamma_6$}
\label{fig:Gamma6}
\end{figure}
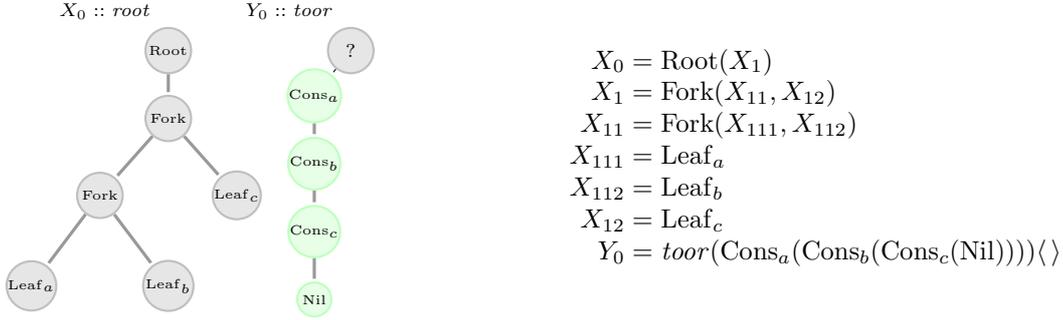
Now the tree rooted at node $X_0$ is closed ~---and thus it no longer holds attributes---~ and the list of its 
leaves has been entirely forwarded to the inherited attribute of node $Y_0$. Note that the recipient node $Y_0$ 
could have been refined in parallel with the changes of configurations just described. 
\ExpleBox{exple:flattening2}
\end{example}

%% file: coprocess.tex
\begin{example}\label{exple:coprocess}
Figure~\ref{fig:coroutines} shows a guarded attribute grammar that represents two 
coroutines communicating through lazy streams. 
Each process alternatively sends and receives data.   
More precisely the second process send an acknowlegment (a $b$ message) upon reception of a message send by the left 
process. Initially or after reception of an acknowlegment of its previous message the left process can either send 
a new message or terminate the communication. 
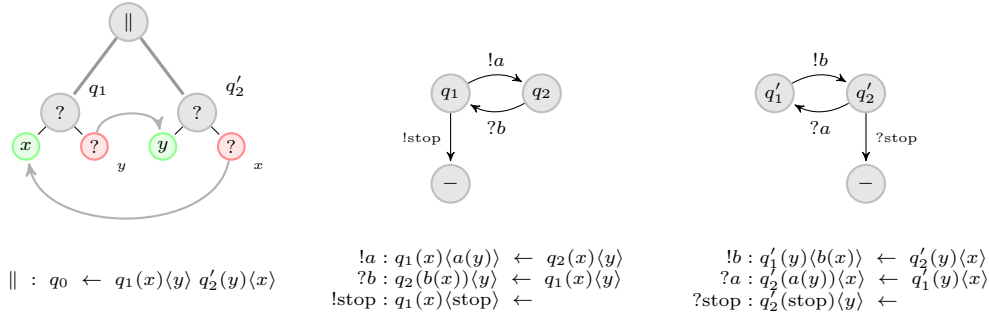
\begin{figure}[htb]
\begin{tabular}{c@{$\quad$}c@{$\quad$}c}
\begin{tikzpicture}[>=stealth',shorten >=1pt, auto,on grid,>=stealth',auto,bend angle=80,
                    every place/.style={minimum size=3.5mm,thick,draw=black!60,inner sep=.2pt},
                    every transition/.style={minimum width=5mm,minimum height=5mm,inner sep=2pt,draw=black!60,thick,fill=white},
                    state/.style={circle,minimum size=5mm,inner sep=1pt,draw=black!25,thick,fill=black!10}]

\begin{scope}[scale=.6]
\node[state]  (x) at (0,2)   {\scriptsize $\|$};
\node[state,label=20:{\scriptsize $q_1$}]  (x1) at (-1.5,0)  {\scriptsize $?$};
\node[state,label=20:{\scriptsize $q'_2$}]  (x2) at (1.5,0) {\scriptsize $?$};
\path[-,color=black!40,very thick] (x)  edge (x1)
                                        edge (x2);

\node[place,draw=green!45,thick,fill=green!10]  (inh1) at (-2.25,-.75)  {\scriptsize $x$};
\node[place,draw=red!45,thick,fill=red!10,label=330:{\tiny $y$}]  (syn1) at (-.75,-.75) {\scriptsize $?$};
\path[-] (x1)  edge (inh1)
               edge (syn1);

\node[place,draw=green!45,thick,fill=green!10]  (inh2) at (.75,-.75)  {\scriptsize $y$};
\node[place,draw=red!45,thick,fill=red!10,label=330:{\tiny $x$}]  (syn2) at (2.25,-.75) {\scriptsize $?$};
\path[-] (x2)  edge (inh2)
               edge (syn2);
               
\path[->,color=black!30,thick] (syn2) edge[bend left] (inh1);
\path[->,color=black!30,thick] (syn1) edge[bend left] (inh2);
\end{scope}
\end{tikzpicture}
&
\begin{tikzpicture}[>=stealth',shorten >=1pt, auto,on grid,>=stealth',auto,bend angle=30,
                    state/.style={circle,minimum size=5mm,inner sep=1pt,draw=black!25,thick,fill=black!10}]
\begin{scope}
 \node at (0,0) {};
\end{scope}

\begin{scope}[scale=.6,yshift=3cm]
\node[state]  (Q1) at (-1,0)   {\scriptsize $q_1$};
\node[state]  (Q2) at (1,0)   {\scriptsize $q_2$};
\node[state]  (Q3) at (-1,-2)   {\scriptsize $-$};
\path[->] (Q1) edge[bend left] node{\scriptsize $!a$} (Q2)
          (Q2) edge[bend left] node{\scriptsize $?b$} (Q1)
          (Q1) edge[swap] node{\tiny $!\mathrm{stop}$} (Q3);
\end{scope}
\end{tikzpicture}
&
\begin{tikzpicture}[>=stealth',shorten >=1pt, auto,on grid,>=stealth',auto,bend angle=30,
                    state/.style={circle,minimum size=5mm,inner sep=1pt,draw=black!25,thick,fill=black!10}]

\begin{scope}
 \node at (0,0) {};
\end{scope}

\begin{scope}[scale=.6,yshift=3cm]
\node[state]  (Q1) at (-1,0)   {\scriptsize $q'_1$};
\node[state]  (Q2) at (1,0)   {\scriptsize $q'_2$};
\node[state]  (Q3) at (1,-2)   {\scriptsize $-$};
\path[->] (Q1) edge[bend left] node{\scriptsize $!b$} (Q2)
          (Q2) edge[bend left] node{\scriptsize $?a$} (Q1)
          (Q2) edge node{\tiny $?\mathrm{stop}$} (Q3);
\end{scope}
\end{tikzpicture}
\\
\begin{scriptsize}
$
\| \;:\; q_0 \;\leftarrow\; q_1(x)\langle y\rangle\; q'_2(y)\langle x\rangle
$
\end{scriptsize}
&
\begin{scriptsize}
$
 \begin{array}{r@{\;:\;}l}
 !a & q_1(x)\langle a(y)\rangle \;\leftarrow\; q_2 (x)\langle y\rangle \\
 ?b & q_2(b(x))\langle y\rangle \;\leftarrow\; q_1 (x)\langle y\rangle \\
 !\mathrm{stop} & q_1(x)\langle \mathrm{stop}\rangle \;\leftarrow\;
 \end{array}
$
\end{scriptsize}
&
\begin{scriptsize}
$
 \begin{array}{r@{\;:\;}l}
 !b & q'_1(y)\langle b(x)\rangle \;\leftarrow\; q'_2 (y)\langle x\rangle \\
 ?a & q'_2(a(y))\langle x\rangle \;\leftarrow\; q'_1 (y)\langle x\rangle \\
 ?\mathrm{stop} & q'_2(\mathrm{stop})\langle y\rangle \;\leftarrow\; 
 \end{array}
$
\end{scriptsize}
\end{tabular}
\caption{Coroutines with lazy streams}
\label{fig:coroutines}
\end{figure}

Production $!a : q_1(x')\langle a(y')\rangle \leftarrow q_2(x')\langle y'\rangle$ 
applies at node $X_1$ of the configuration 
\begin{small}
\[
 \Gamma_1 = \set{X=X_1\| X_2, \;\; X_1=q_1(x)\langle y\rangle,\;\; X_2=q'_2(y)\langle x\rangle}
\]
\end{small}
shown on the left of Figure~\ref{fig:coroutines2} 
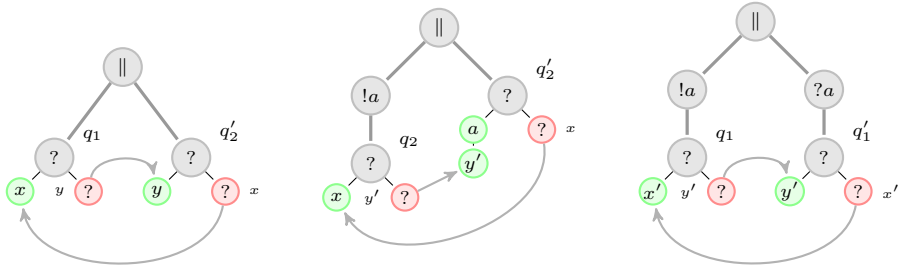
\begin{figure}[htb]
\begin{tabular}{c@{$\qquad$}c@{$\qquad$}c}
\begin{tikzpicture}[>=stealth',shorten >=1pt, auto,on grid,>=stealth',auto,bend angle=80,
                    every place/.style={minimum size=3.5mm,thick,draw=black!60,inner sep=.2pt},
                    every transition/.style={minimum width=5mm,minimum height=5mm,inner sep=2pt,draw=black!60,thick,fill=white},
                    state/.style={circle,minimum size=5mm,inner sep=1pt,draw=black!25,thick,fill=black!10}]

\begin{scope}[scale=.6]
\node[state]  (x) at (0,2)   {\scriptsize $\|$};
\node[state,label=20:{\scriptsize $q_1$}]  (x1) at (-1.5,0)  {\scriptsize $?$};
\node[state,label=20:{\scriptsize $q'_2$}]  (x2) at (1.5,0) {\scriptsize $?$};
\path[-,color=black!40,very thick] (x)  edge (x1)
                                        edge (x2);

\node[place,draw=green!45,thick,fill=green!10]  (inh1) at (-2.25,-.75)  {\scriptsize $x$};
\node[place,draw=red!45,thick,fill=red!10,label=left:{\tiny $y$}]  (syn1) at (-.75,-.75) {\scriptsize $?$};
\path[-] (x1)  edge (inh1)
               edge (syn1);

\node[place,draw=green!45,thick,fill=green!10]  (inh2) at (.75,-.75)  {\scriptsize $y$};
\node[place,draw=red!45,thick,fill=red!10,label=right:{\tiny $x$}]  (syn2) at (2.25,-.75) {\scriptsize $?$};
\path[-] (x2)  edge (inh2)
               edge (syn2);
               
\path[->,color=black!30,thick] (syn2) edge[bend left] (inh1);
\path[->,color=black!30,thick] (syn1) edge[bend left] (inh2);
\end{scope}
\end{tikzpicture}
&
\begin{tikzpicture}[>=stealth',shorten >=1pt, auto,on grid,>=stealth',auto,bend angle=80,
                    every place/.style={minimum size=3.5mm,thick,draw=black!60,inner sep=.2pt},
                    every transition/.style={minimum width=5mm,minimum height=5mm,inner sep=2pt,draw=black!60,thick,fill=white},
                    state/.style={circle,minimum size=5mm,inner sep=1pt,draw=black!25,thick,fill=black!10}]

\begin{scope}[scale=.6]
\node[state]  (x) at (0,1.5)   {\scriptsize $\|$};
\node[state]  (x1) at (-1.5,0)  {\scriptsize $!a$};
\node[state,label=20:{\scriptsize $q'_2$}]  (x2) at (1.5,0) {\scriptsize $?$};
\node[state,label=20:{\scriptsize $q_2$}]  (x3) at (-1.5,-1.5) {\scriptsize $?$};
\path[-,color=black!40,very thick] (x)  edge (x1)
                                        edge (x2)
                                   (x1) edge (x3);

\node[place,draw=green!45,thick,fill=green!10]  (inh1) at (-2.25,-2.25)  {\scriptsize $x$};
\node[place,draw=red!45,thick,fill=red!10,label=left:{\tiny $y'$}]  (syn1) at (-.75,-2.25) {\scriptsize $?$};
\path[-] (x3)  edge (inh1)
               edge (syn1);

\node[place,draw=green!45,thick,fill=green!10]  (inh2) at (.75,-.75)  {\scriptsize $a$};
\node[place,draw=green!45,thick,fill=green!10]  (inh21) at (.75,-1.5)  {\scriptsize $y'$};
\node[place,draw=red!45,thick,fill=red!10,label=right:{\tiny $x$}]  (syn2) at (2.25,-.75) {\scriptsize $?$};
\path[-] (x2)  edge (inh2)
               edge (syn2)
          (inh2) edge (inh21);
               
\path[->,color=black!30,thick] (syn2) edge[bend left] (inh1);
\path[->,color=black!30,thick] (syn1) edge (inh21);
\end{scope}
\end{tikzpicture}
&
\begin{tikzpicture}[>=stealth',shorten >=1pt, auto,on grid,>=stealth',auto,bend angle=80,
                    every place/.style={minimum size=3.5mm,thick,draw=black!60,inner sep=.2pt},
                    every transition/.style={minimum width=5mm,minimum height=5mm,inner sep=2pt,draw=black!60,thick,fill=white},
                    state/.style={circle,minimum size=5mm,inner sep=1pt,draw=black!25,thick,fill=black!10}]

\begin{scope}[scale=.6]
\node[state]  (x) at (0,1.5)   {\scriptsize $\|$};
\node[state]  (x1) at (-1.5,0)  {\scriptsize $!a$};
\node[state]  (x2) at (1.5,0) {\scriptsize $?a$};
\node[state,label=20:{\scriptsize $q_1$}]  (x3) at (-1.5,-1.5)  {\scriptsize $?$};
\node[state,label=20:{\scriptsize $q'_1$}]  (x4) at (1.5,-1.5) {\scriptsize $?$};
\path[-,color=black!40,very thick] (x)  edge (x1)
                                        edge (x2)
                                    (x1) edge (x3)
                                    (x2) edge (x4);

\node[place,draw=green!45,thick,fill=green!10]  (inh1) at (-2.25,-2.25)  {\scriptsize $x'$};
\node[place,draw=red!45,thick,fill=red!10,label=left:{\tiny $y'$}]  (syn1) at (-.75,-2.25) {\scriptsize $?$};
\path[-] (x3)  edge (inh1)
               edge (syn1);

\node[place,draw=green!45,thick,fill=green!10]  (inh2) at (.75,-2.25)  {\scriptsize $y'$};
\node[place,draw=red!45,thick,fill=red!10,label=right:{\tiny $x'$}]  (syn2) at (2.25,-2.25) {\scriptsize $?$};
\path[-] (x4)  edge (inh2)
               edge (syn2);
               
\path[->,color=black!30,thick] (syn2) edge[bend left] (inh1);
\path[->,color=black!30,thick] (syn1) edge[bend left] (inh2);
\end{scope}
\end{tikzpicture}
\end{tabular}
\caption{\protect{$\netfir{\Gamma_1}{!a/X_1}{}\netfir{\Gamma_2}{?a/X_2}{\Gamma_3}$}}
\label{fig:coroutines2}
\end{figure}
because its left-hand side 
$q_1(x')\langle a(y')\rangle$ matches with the definition 
$q_1(x)\langle y\rangle$ of $X_1$ with $\sigma_{\mathit{in}}=\set{x'=x}$ and 
$\sigma_{\mathit{out}}=\set{y=a(y')}$. One obtains the configuration 
\begin{small}
\[
 \Gamma_2 = \set{X=X_1\| X_2,\;\; X_1=!a(X_{11}),\;\; X_2= q'_2(a(y'))\langle x\rangle, \;\; X_{11}=q_2(x)\langle y'\rangle}
\]
\end{small}
shown on the middle of Figure~\ref{fig:coroutines2}. 
Production $?a:q'_2(a(y))\langle x'\rangle \leftarrow q'_1(y)\langle x'\rangle$ applies at node $X_{2}$ of $\Gamma_2$ 
because its left-hand side $q'_2(a(y))\langle x'\rangle$ matches with the definition 
$q'_2(a(y'))\langle x\rangle$ of $X_2$ with $\sigma_{\mathit{in}}=\set{y=y'}$ and 
$\sigma_{\mathit{out}}=\set{x=x'}$. One obtains the configuration 
\begin{small}
\[
 \Gamma_3 = \set{X=X_1\| X_2,\;\; X_1=!a(X_{11}),\;\; X_2= ?a(X_{21}),\;\; X_{11}=q_2(x')\langle y'\rangle, \;\; X_{21}=q'_1(y')\langle x'\rangle}
\]
\end{small}
shown on the right of Figure~\ref{fig:coroutines2}. 
The corresponding acknowlegment may then be send and received leading to configuration 
\begin{small}
\[
 \Gamma_5 = \Gamma \cup \set{X_{111}=q_1(x)\langle y\rangle, \;\; X_{211}=q'_2(y)\langle x\rangle}
\]
\end{small}
where 
\begin{scriptsize}
$\Gamma=\set{X=X_1\| X_2,\;\; X_1=!a(X_{11}),\;\; X_2= ?a(X_{21}),\;\;X_{21}=!b(X_{211}),\;\; X_{11}= ?b(X_{111})}$.
\end{scriptsize}
The process on the left may decide to end the communication by applying production 
$!\mathrm{stop} : q_1(x')\langle \mathrm{stop}\rangle \leftarrow\;$ at $X_{111}$  
with $\sigma_{\mathit{in}}=\set{x'=x}$ and $\sigma_{\mathit{out}}=\set{y=\mathrm{stop}}$ leading to configuration  
\begin{small}
\[
 \Gamma_6 = \Gamma \cup \set{X_{111}=!\mathrm{stop}, \;\; X_{211}=q'_2(\mathrm{stop})\langle x\rangle}
\]
\end{small}
The reception of this message by the process on the right corresponds to applying production
$?\mathrm{stop} : q'_2(\mathrm{stop})\langle y\rangle \leftarrow\;$ at $X_{211}$ with 
$\sigma_{\mathit{in}}=\emptyset$ and $\sigma_{\mathit{out}}=\set{x=y}$ leading to configuration
\begin{small}
\[
 \Gamma_7 = \Gamma \cup \set{X_{111}=!\mathrm{stop}, \;\; X_{211}=?\mathrm{stop}}
\]
Note that variable $x$ appears in an input position in $\Gamma_6$ and has no corresponding output occurrence. This means that 
the value of $x$ is not used in the configuration. When production $?\mathrm{stop}$ is fired in node $X_{211}$ variable $y$ 
is substituted to $x$. Variable $y$ has an output occurrence in production $?\mathrm{stop}$ and no input occurence meaning that the 
corresponding output attribute is not defined by the semantic rules. As a consequence this variable simply disappears in the 
resulting configuration $\Gamma_7$. If variable $x$ was used in $\Gamma_6$ then the output occurrences of $x$ would have been replaced by 
(output occurrences) of variable $y$ which will remain undefined (no value will be substituted to $y$ in subsequent transformations) 
until these occurrences of variables may possibly disappear. 
\end{small}
\ExpleBox{exple:coprocess}
\end{example}

%% file: distribution.tex
The fact that triggered productions are not enabled  can also impact the distributability of a grammar as shown by the following example.

\begin{example}
\label{exple:nonDGAG}
Let us consider the GAG with the following \vspace{-\medskipamount}productions:
\[
\begin{array}{c@{\,:\;}l@{\;\leftarrow\;}l}
P & s(\,)\langle\,\rangle & s_1(x)\langle y\rangle\;\; s_2(y)\langle x\rangle\\
Q & s_1(z)\langle a(z)\rangle & \\
R & s_2(u)\langle a(u)\rangle
\end{array}
\vspace{-\medskipamount}
\]
Production $P$ is enabled in the configuration $\Gamma_0=\set{X_0=s(\,)\langle\,\rangle}$ with $\netfir{\Gamma_0}{P/X_0}{\Gamma_1}$ where 
$\Gamma_1=\set{X_0=P(X_1,X_2);\,X_1=s_1(x)\langle y\rangle,\,X_2=s_2(y)\langle x\rangle}$. 
In configuration $\Gamma_1$ productions $Q$ and $R$ are enabled in nodes $X_1$ and $X_2$ respectively \vspace{-\medskipamount}with 
\[
\begin{array}{l}
\netfir{\Gamma_1}{Q/X_1}{\Gamma_2}\;\;\mathrm{where}\;\;\Gamma_2=\set{X_0=P(X_1,X_2);\,X_1=Q,\,X_2=s_2(a(x))\langle x\rangle}\\
\netfir{\Gamma_1}{R/X_2}{\Gamma_3}\;\;\mathrm{where}\;\;\Gamma_3=\set{X_0=P(X_1,X_2);\,X_1=s_2(a(y))\langle y\rangle,\,X_2=R}
\end{array}
\]
Now production $R$  is triggered but not enabled in node $X_2$  configuration $\Gamma_2$  because of the cyclicity of $\set{x=a(a(x))}$.
Similarly, production $Q$  is triggered but not enabled in node $X_3$  configuration $\Gamma_3$. 
There is a conflict between the application of productions $R$ and $Q$ in configuration $\Gamma_1$,  which makes this specification non-implementable in case nodes 
$X_1$ and $X_2$ have distinct locations. 
\ExpleBox{exple:nonDGAG}
\end{example}

\subsection{Input Enabled Guarded Attribute Grammars}
Substitution $\sigma_{\mathit{in}}$, given by pattern matching, is monotonous w.r.t. incoming information and thus it causes no problem for a 
distributed implementation of a model. However substitution $\sigma_{\mathit{out}}$ is not monotonous since it may become 
undefined when information coming from a distant location makes 
the match of output attributes a cyclic set of equations, 
as illustrated by example~\ref{exple:nonDGAG}.

\begin{definition}
\label{def:WFGAG}
A guarded attribute grammar is \textbf{input-enabled}  
if every production that is triggered in an accessible configuration is also enabled.  
\DefBox{def:WFGAG}
\end{definition}

We call the \textit{substitution induced by a sequence} $\netfir{\Gamma}{*}{\Gamma'}$ the corresponding composition of the 
various substitutions associated respectively with each of the individual steps in the sequence. 
If $X$ is an open node in both $\Gamma$ and $\Gamma'$, i.e.,\ 
no productions are applied at $X$ in the sequence, then we get $X=s(d_1\sigma,\ldots,d_n\sigma)\langle y_1,\ldots,y_m\rangle\in\Gamma'$ 
where $X=s(d_1,\ldots,d_n)\langle y_1,\ldots,y_m\rangle\in\Gamma$ and $\sigma$ is the substitution induced by the sequence. 

\begin{proposition}[Monotony]
\label{prop:monotony}
Let $\Gamma$ be an accessible configuration of an input-enabled GAG, 
%guarded attribute grammar, 
$X=s(d_1,\ldots,d_n)\langle y_1,\ldots,y_m\rangle\in\Gamma$ and $\sigma$ the substitution induced by some sequence 
starting from $\Gamma$. 
Then $\netfir{\Gamma}{P/X}{\Gamma'}$ implies $\netfir{\Gamma\sigma}{P/X}{\Gamma'\sigma}$.
\end{proposition}
\begin{proof}[Proof of Proposition~\ref{prop:monotony}]\mbox{}\\
Direct consequence of Definition~\ref{def:GAG} due to the fact that 
\begin{enumerate}
 \item $\mathbf{match}(p,d\sigma) =\mathbf{match}(p,d)\sigma$, and 
 \item $\mathbf{mgu}(\setof{y_j=u_j\sigma}{1\leq j\leq m})=\mathbf{mgu}(\setof{y_j=u_j}{1\leq j\leq m})\sigma$.
\end{enumerate}
The former is trivial and the latter follows by induction on the length of the computation of the most general unifier (relation $\Rightarrow^*$ using 
rule~(5) only). Note that the assumption that the guarded attribute grammar is input-enabled is crucial because in the general case it could happen that 
the set $\setof{y_j=u_j\sigma_{\mathit{in}}}{1\leq j\leq m}$ satisfies the occur check whereas the set 
$\setof{y_j=u_j(\sigma_{\mathit{in}}\sigma)}{1\leq j\leq m}$   does not satisfy the occur check. 
\PropBox{prop:monotony}
\end{proof}

Proposition~\ref{prop:monotony} is instrumental for the distributed implementation of guarded attribute grammars. 
Namely it states that new information coming from a distant 
asynchronous location refining the value of some input occurrences of variables of an enabled production do not 
prevent from applying that production. Thus a production that is locally enabled can  freely be applied 
regardless of information that might further refine the current local configuration. 
It means that conflict arises only from the existence of two distinct productions enabled in the same 
open node. Hence the only form of non-determinism corresponds to the decision of a stakeholder to apply 
one particular production among those enabled in a configuration. This is expressed by the following confluence property. 

\vspace{-\medskipamount}
\begin{corollary}
\label{cor:confluence}
Let $\Gamma$ be an accessible configuration of an input enabled GAG. 
If $\netfir{\Gamma}{P/X}{\Gamma_1}$ and $\netfir{\Gamma}{Q/Y}{\Gamma_2}$ with $X\neq Y$ then 
$\netfir{\Gamma_2}{P/X}{\Gamma_3}$ and $\netfir{\Gamma_1}{Q/Y}{\Gamma_3}$ for some configuration $\Gamma_3$. 
\end{corollary}

Note that, by Corollary~\ref{cor:confluence}, the artifact contains a full history of the case in the sense that one can reconstruct from the artifact 
the complete sequence of applications of productions leading to the resolution of the case (up to the commutation of independent elements in the sequence).

We might have considered a more symmetrical presentation in Definition~\ref{def:GAG}  by allowing patterns 
for synthesized attributes in the right-hand sides of productions  with the effect of 
creating forms in a configuration with patterns in their co-arguments. These patterns express constraints on synthesized values.
This extension could be acceptable as long as one sticks to  purely centralized models. 
However, as soon as one wants to distribute the model on an asynchronous architecture, one cannot avoid such a constraint to be further refined 
due to a transformation occurring in a distant location. Then the monotony property (Proposition~\ref{prop:monotony}) is lost: a locally enabled 
production can later be disabled when a constraint on a synthesized value gets a refined value. 
This is why we required synthesized attributes in the right-hand side of a production to be given by plain variables in order to prohibit 
the expression of constraints on synthesized values. 

It is difficult to verify input-enabledness as the whole set of accessible configurations are involved in this condition. 
Nevertheless one can find a sufficient condition for input enabledness, similar to the strong non-circularity of 
attribute grammars \cite{CourcelleF82}, that can be checked by a simple fixed-point computation. 

\begin{definition}\label{def:acyclic}
Let $s$ be a sort of a guarded attribute grammar with $n$ inherited attributes and $m$ synthesized attributes. 
We let $(j,i)\in SI(s)$ where $1\leq i\leq n$ and $1\leq j\leq m$ if exists 
$X=s(d_1,\ldots,d_n)\langle y_1,\ldots, y_m\rangle\in\Gamma$ where $\Gamma$ is an accessible configuration and 
$y_j\in d_i$. 
If $P$ is a production with left-hand side $s(p_1,\ldots,p_n)\langle u_1,\ldots,u_m\rangle$ we let 
$(i,j)\in IS(P)$ if exists a variable $x\in \mathrm{var}(P)$ such that $x\in \mathrm{var}(d_i)\cap\mathrm{var}(u_j)$.
The guarded attribute grammar $G$ is said to be \textbf{acyclic} if for every sort $s$ and production $P$  whose 
left-hand side is a form of sort $s$ the graph $G(s,P)=SI(s)\cup IS(P)$ is acyclic. 
\DefBox{def:acyclic}
\end{definition}
\begin{proposition}\label{prop:acyclic}
An acyclic guarded attribute grammar is input-enabled.
\end{proposition}
\begin{proof}
Suppose $P$ is triggered in node $X$ with substitution $\sigma_{\mathit{in}}$ such that $y_j\in u_i\sigma_{\mathit{in}}$ then $(i,j)\in G(s,P)$. 
Then the fact that occur check fails for the set $\setof{y_j}{1\leq j\leq m}$ entails that one can find a cycle in $G(s,P)$.
\PropBox{prop:acyclic}
\end{proof}

Relation $SI(s)$ still takes into account the whole set of accessible configurations.
The following definition provides an overapproximation of this relation given  
by a fixed point computation. 

\begin{definition}\label{def:GDL}
The \textbf{graph of local dependencies} of a production $P:F_0\leftarrow F_1\cdots F_{\ell}$  is the directed graph $GLD(P)$ that records the data dependencies 
between the occurrences of attributes given by the semantics rules. We designate the occurrences of attributes of $P$ as follows: 
we let $k(i)$ (respectively $k\langle j\rangle$) 
denote the occurrence of the $i^{th}$ inherited attribute (resp. the $j^{th}$ synthesized attribute) in $F_k$. 
If $s$ is a sort  with $n$ inherited attributes and $m$ synthesized attributes we define the relations $\overline{IS(s)}$ and $\overline{SI(s)}$ over 
$[1,n]\times [1,m]$ and $[1,m]\times [1,n]$ respectively as the least relations such that :
\begin{enumerate}
 \item $\overline{SI(s)}=SI(s)$ if $s$ is an axiom, i.e.,\ it is given by the set of pairs $(j,i)$ such that $y_j\in \mathrm{var}(d_i)$ for some service 
 $F=s(d_1,\ldots,d_n)\langle y_1,\ldots,y_m\rangle$ of sort $s$ in the interface of the guarded attribute grammar.
 \item For every production $P:F_0\leftarrow F_1\cdots F_{\ell}$ where form $F_i$ is of sort $s_i$ and for every $k\in [1,\ell]$ 
 \[
  \setof{(j,i)}{(k\langle j\rangle,k(i))\in GLD(P)^{k}} \subseteq \overline{SI(s_k)}
 \]
where graph $GLD(P)^{k}$ is given as the transitive closure of
\[
 \begin{array}{c}
 GLD(P) \cup  \setof{(0\langle j\rangle,0(i))}{(j,i)\in \overline{SI(s_{0})}} \\
 \cup \setof{(k'(i),k'\langle j\rangle)}{k'\in [1,\ell],\; k'\neq k,\; (i,j)\in \overline{IS(s_{k'})}}
  
  \end{array}
\]
 \item For every production $P:F_0\leftarrow F_1\cdots F_{\ell}$ where form $F_i$ is of sort $s_i$ 
 \[
  \setof{(i,j)}{(0(i),0\langle j\rangle)\in GLD(P)^{0}} \subseteq \overline{IS(s_0)}
 \]
where graph $GLD(P)^{0}$ is given as the transitive closure of
\[
 GLD(P) \cup \setof{(k(i),k\langle j\rangle)}{k\in [1,\ell],\; (i,j)\in \overline{IS(s_{k})}}
\]
\end{enumerate}
The guarded attribute grammar $G$ is said to be \textbf{strongly-acyclic} if for every sort $s$ and production $P$  whose 
left-hand side is a form of sort $s$ the graph $\overline{G(s,P)}=\overline{SI(s)}\cup IS(P)$ is acyclic. 
\DefBox{def:GDL}
\end{definition}

\begin{proposition}\label{prop:stronglyAcyclic}
A strongly-acyclic guarded attribute grammar is acyclic and hence input-enabled. 
\end{proposition}
\begin{proof}
The proof is analog to the proof that a strongly non-circular attribute grammar is non-circular and it goes as follows. 
We let $(i,j)\in IS(s)$ when $\mathrm{var}(d_i\sigma)\cap\mathrm{var}(y_j\sigma)\neq\emptyset$ for some form  
$F=s(d_1,\ldots,d_n)\langle y_1,\ldots, y_m\rangle$ of sort $s$ and where $\sigma$ is the substitution induced 
by a firing sequence starting from configuration $\set{X=F}$. Then we show by induction on the length 
of the firing sequence leading to the accessible configuration that $IS(s)\subseteq \overline{IS(s)}$ and $SI(s)\subseteq \overline{SI(s)}$. 
\PropBox{prop:stronglyAcyclic}
\end{proof}

Note that the following two inclusions are strict
\[
 \mbox{strongly-acyclic GAGs} \subsetneq \mbox{acyclic GAGs} \subsetneq \mbox{input enabled GAGs}
\]
Indeed the reader may easily check that the guarded attribute grammar 
\[
 \left\{
 \begin{array}{l}
  A(x)\langle z\rangle \leftarrow B(a(x,y))\langle y,z\rangle \\
  B(a(x,y))\langle x,y\rangle \leftarrow 
 \end{array}
 \right.
\]
is cyclic and input-enabled 
whereas guarded attribute grammar with productions 
\[
 \left\{
 \begin{array}{l}
  A(x)\langle z\rangle \leftarrow B(y,x)\langle z,y\rangle \\
  A(x)\langle z\rangle \leftarrow B(x,y)\langle y,z\rangle \\
  B(x,y)\langle x,y\rangle \leftarrow 
 \end{array}
 \right.
\]
is acyclic but not strongly-acyclic. 
Attribute grammars arising from real situations are almost always strongly non-circular so that this assumption is not really restrictive.
Similarly we are confident that most of the guarded attribute grammars that we shall use in practise will be input-enabled 
and that most of the input-enabled guarded attribute grammars are in fact strongly-acyclic. 
Thus most of the specifications are distributable and most of those can be proved so by checking the strong non-circularity condition. 

\subsection{Distribution of an Input Enabled Guarded Attribute Grammar}
The principle of a distribution of a GAG on a set of locations is as follows: 
Each location maintains a local configuration, and subscribes to results provided by other locations. 
Productions are applied locally. 
When variables are given a value by a production, the location that computed this value sends messages to the locations that subscribed to this value. 
Messages are simply equations defining the value of a particular variable.   
Upon reception of a messages, a subscriber updates its local configuration, and may in turn produce new messages.

More formally, a GAG can be distributed by specifying a partition $S=\uplus_{1\leq \ell\leq p}S_{\ell}$ of the set of sorts. The projections $\Gamma_{\ell}$, 
called the local configurations associated with sites $S_{\ell}$, are defined as follows. Each site $S_{\ell}$ has a namespace $ns(S_{\ell})$ used 
for the nodes $X$ whose sorts are in $S_{\ell}$ and for the variables $x$ representing attributes of these nodes but also for references to variables 
belonging to distant sites (subscriptions). Hence we have name generators that produce unique identifiers for each newly created 
variable for each site. For each equation $X=P(X_1,\ldots, X_n)$ with $X::s$ and $X_i::s_i$ we insert equation 
$X=P(\overline{X}_1,\ldots, \overline{X}_n)$ in $\Gamma_{\ell}$ where $s\in S_{\ell}$ and variable $\overline{X}_i$ is $X_i$ if $s_i\in S_{\ell}$ or is a new variable 
in the namespace of $S_{\ell}$ if $s_i\in S_{\ell'}$  with $\ell'\neq \ell$. In the latter case we add equation $\overline{X}_i=X_i$ in $\Gamma_{\ell}$. 
Similarly for each equation $X=s(t_1,\ldots, t_n)\langle y_1,\ldots,y_m\rangle$ in $\Gamma$ we add equation 
$X=s(\overline{t}_1,\ldots, \overline{t}_n)\langle \overline{y}_1,\ldots,\overline{y}_m\rangle$ in $\Gamma_{\ell}$ where $s\in S_{\ell}$ and 
$\overline{t}$ is obtained by replacing each variable $x$ in term $t$ by $\overline{x}$ where variable $\overline{x}$ is $x$ if $x::s'$ with $s'\in S_{\ell}$ else is a 
new variable in the namespace of $S_{\ell}$. In the latter case one adds equation $\overline{x}=x$, called 
a {\bf subscription}, to $\Gamma_{\ell'}$. Similarly for the variables $y_j$. Hence a local configuration contains the usual equations associated with their closed and open nodes (and containing only local variables) 
together with equations of the form $X=Y$ and $y=x$ where $x$ and $X$ are local names and $y$ and $Y$ belongs to distant sites. 
Clearly the global configuration can be recovered as $\Gamma=\Gamma_1\oplus \cdots \oplus \Gamma_n$ where operator $\oplus$ consists in taking the union of the systems 
of equations given as arguments and simplifying the resulting system by elimination of the copy rules: we drop each equation of the form $X=Y$ (respectively $y=x$) 
and replace each occurrence of $X$ by $Y$ (resp. of $y$ by $x$). Therefore the global configuration $\Gamma$ may be identified with the vectors of local configurations 
$(\Gamma_1,\ldots,\Gamma_p)$.

Each production can then be locally applied: we write $\lsta{\Gamma_{\ell}}{P/X}{M}{\Gamma'_{\ell}}$  
when application of production $P$ at node $X$ results in a new configuration $\Gamma'_{\ell}$ 
and the sending of a set of messages $M$. 

More formally, $\lsta{\Gamma_{\ell}}{P/X}{M}{\Gamma'_{\ell}}$  when $X=s(t_1,\ldots,t_n)\langle y_1,\ldots, y_m\rangle\in \Gamma_{\ell}$ 
and $P=F\leftarrow F_1\cdots F_k$ is a production whose left-hand side matches with $X$ and 
\[
\begin{array}{l}
  \begin{array}{rcl}
   \Gamma'_{\ell} & = & \set{X=P(X_1,\ldots,X_k)}\qquad \mbox{where}\;\;X_1,\ldots,X_k\;\;\mbox{are new names in $ns(S_{\ell})$}\\
           & \cup & \setof{X_i=F_i\overline{\sigma}}{X_i::s_i\;\;\mbox{and}\;\; s_i\in S_{\ell}} \\
           & \cup & \setof{X'=F\overline{\sigma}}{(X'=F)\in\Gamma_{\ell}\;\wedge\; X'\neq X}\\
           & \cup & \setof{y'=y_j\sigma}{(y'=y_j)\in \Gamma_{\ell}\;\;\mbox{and $y_j\sigma$ is a variable}}
  \end{array}
\\
  \begin{array}{rcl}
   M & = & \setof{X_i=F_i\sigma}{X_i::s_i\;\;\mbox{and}\;\; s_i\not\in S_{\ell}} \\
           & \cup & \setof{y'=y_j\sigma}{(y'=y_j)\in \Gamma_{\ell}\;\;\mbox{and $y_j\sigma$ not a variable}}\\
           & \cup & M_\sigma
  \end{array}
\end{array}
\]
where $\sigma=\mathbf{match}(F,X)$, and $\overline{\sigma}$ is the relativization of $\sigma$ to location $\ell$, it generates a set of messages $M_\sigma$.
This relation means that applying production $P$ at $X$ in site $S_{\ell}$ generates messages $M$ send to distant sites. 
The reception of a message may generate new messages and is described by relation $\lsta{\Gamma_{\ell}}{m}{M}{\Gamma'_\ell}$ where 
\vspace{-\medskipamount}
\begin{enumerate}
 \item If $m=\{ X=s(t_1,\ldots,t_n)\langle y_1,\ldots,y_q\rangle\}$ with $X\in ns(S_{\ell'})$, $s\in S_{\ell}$ with $\ell'\neq \ell$ then 
 $\Gamma'_{\ell}=\Gamma_{\ell}\cup\set{\overline{X}=s(\overline{t}_1,\ldots, \overline{t}_n)\langle \overline{y}_1,\ldots,\overline{y}_q\rangle} 
 \cup\setof{y_j=\overline{y_j}}{1\leq j\leq q}$ where $\overline{X}$, the variables $\overline{x}$ for $x\in var(t_i)$ and the variables $\overline{y}_j$ 
 are new names in $ns(S_{\ell})$ and $\overline{t}=t[\overline{x}/x]$, and $M=\setof{\overline{x}=x}{x\in var(t_i)\;\; 1\leq i\leq n}\cup\set{X=\overline{X}}$.

 \item If $m=\{x=t\}$ with $x\in ns(S_{\ell})$ then $\Gamma'_{\ell}=\Gamma_{\ell}[x=t[\overline{y}/y]]$ where $\overline{y}$ are new names in $ns(S_{\ell})$ 
 associated with the variables $y$ in $t$ and $M=\setof{\overline{y}=y}{y\in var(t)}$.
 \item If $m=(X=Y)$ with $X\in ns(S_{\ell})$ then $\Gamma'_{\ell}=\Gamma_{\ell}\cup\set{X=Y}$ and $M=\emptyset$.
 \item If $m=(y=x)$ with $x\in ns(S_{\ell})$ then $\Gamma'_{\ell}=\Gamma_{\ell}\cup\set{y=x}$ and $M=\emptyset$.
\end{enumerate}

The global dynamics of the system can then be derived as follows, where $e$ stands for $P/X$ or a message \vspace{-2mm}$m$:
\begin{enumerate}
 \item If $\lsta{\Gamma_{\ell}}{e}{M}{\Gamma'_{\ell}}$ then $\gsta{\Gamma}{e}{M}{\Gamma'}$ with $\Gamma_{\ell'}=\Gamma'_{\ell'}$ for $\ell'\neq\ell$.
 \item If $\gsta{\Gamma}{e}{M}{\Gamma'}$ and $\gsta{\Gamma'}{m}{M'}{\Gamma''}$ for $m\in M$ then $\gsta{\Gamma}{e}{M\setminus\set{m}\cup M'}{\Gamma''}$
\end{enumerate}

Input-enabled GAGs possess useful properties with respect to distribution, 
namely messages consumptions and application of productions commute, as shown in the following proposition:

\begin{proposition}
\label{prop:distribution}
For an input-enabled guarded attribute \vspace{-1mm}grammar: 
\begin{enumerate}
 \item If $\gsta{\Gamma}{P/X}{M}{\Gamma'}$ then there exists a substitution $\sigma_{M}$ such that $\gsta{\Gamma}{P/X}{\emptyset}{\Gamma'\sigma_M}$. 
 \item $\netfir{\Gamma}{P/X}{\Gamma'}$ if and only if $\gsta{\Gamma}{P/X}{\emptyset}{\Gamma'}$
 \item Let $\gsta{\Gamma}{P_1/X_1}{M_1}{\Gamma_1}$ and $\gsta{\Gamma}{P_2/X_2}{M_2}{\Gamma_2}$ with $X_1\neq X_2$. 
 One can assume w.l.o.g that  $M_1$ and $M_2$ have no common variables (the name generator chooses different names for the new variables 
 in both cases). Then the diagram below , where $\rightsquigarrow$ denotes messages consumption commutes. 
\end{enumerate}

\end{proposition}
\newcommand{\testlabel}{{\begin{array}{ll}
                &\Gamma_3(\sigma_{M_1}\cup\sigma_{M_2}) \\
                =&\Gamma_2'\sigma_{M_1} = \Gamma_1'\sigma_{M_2}\\
                \end{array}}
                }
\vspace{-2\bigskipamount}
\begin{displaymath}
    \xymatrix{ & & \Gamma \ar[dl]^{M_1}_{P_1/X_1}\ar[dr]^{P_2/X_2}_{M_2} &  & \\
               \Gamma_1\sigma_{M_1} \ar[d]^{P_2/X_2}_{M_2} & \Gamma_1 \ar@{~>}[l]\ar[dr]^{P_2/X_2}_{M_2} &  & \Gamma_2\ar@{~>}[r]\ar[dl]^{M_1}_{P_1/X_1} & \Gamma_2\sigma_{M_2}\ar[d]^{P_1/X_1}_{M_1}\\
               \Gamma_1' \ar@{~>}[drr] &  & \Gamma_3\ar@{~>}[d] & & \Gamma_2'\ar@{~>}[dll]\\
                &   & \testlabel&  & }
\end{displaymath}

Intuitively, proposition~\ref{prop:distribution} and in particular (3) mean that distribution does not affect the global behavious of an input-enabled GAG. 
\begin{proof}[Proof of Proposition~\ref{prop:distribution}]\mbox{}\\
We first prove $(1)$: whenever $\gsta{\Gamma}{P/X}{M}{\Gamma'}$ then there exists a substitution $\sigma_{M}$ such that $\gsta{\Gamma}{P/X}{\emptyset}{\Gamma'\sigma_M}$.
Let us assume that $\gsta{\Gamma}{P/X}{M}{\Gamma'}$, and examine how consuming messages in $M=\{m_1,\dots,m_q\}$ affects $\Gamma'$. Messages can be of several kinds : 
\begin{itemize} 
\item if $m_i=(y=x)$ ( or $m_i=(X=Y)$ ) then consuming the message results in adding an equation $\sigma_{m_i}=\{y=x\}$ (resp. $\sigma_{m_i}=\{X=Y\}$) to the local configuration that receives this message, and generates no new message. 
\item if $m_i= \{ x=t\}$, then consumption of the message results in production of new variables, and a new (finite) set of messages $M_i$ that are all of the form $\{\bar{y}=y\}$ and can then be consumed without producing new messages by the location that has subscribed to this value. We can denote by $\sigma_i$ the substitution that replaces every $x$ in the local configuration that receives $m_i$.
\item if $m_i$ is of the form $\{X=s(t_1,\dots, t_n)<y_1,\dots,y_n>\}$, then consuming $m_i$ results in adding new equations to the local configuration that receives it, and generating a set of messages $M_i=\{m_{i,1}, m_{i,q}\}$, that are of the form $\{\bar{y}=y\}$ and $\{X=\bar{X}\}$ and can hence be consumed by the location that will receive them without generating new messages.
\end{itemize}
These observations show that, after application of a production, message consumption  is a finite process. 
We have $\gsta{\Gamma}{P/X}{M}{\Gamma'}\gsta{}{m_1}{M_1}{\Gamma_1}\gsta{}{M_1}{\emptyset}{\Gamma_1'}\dots 
\gsta{}{m_{|M|}}{M_{|M|}}{\Gamma_{|M|}}\gsta{}{M_{|M|}}{\emptyset}{\Gamma_{|M|}'}$ 

Now, the difference between each $\Gamma_i$ and $\Gamma'_i$ is a set of subscriptions, 
that are appended to some local configurations, and erased during the step. 
Therefore, the global configurations $\Gamma_i$ and $\Gamma'_i$ are identical. 
Similarly, we have $\Gamma_1 = \Gamma' \sigma_{m_1}$, and $\Gamma_i = \Gamma_{i-1}' \sigma_{m_i}$ for every $i\in [2,|M|]$.
Hence, the substitution $\sigma_{M}=\sigma_{m_1}\ldots\sigma_{|M|}$ is such that $\Gamma'\sigma_M = \Gamma_{|M|}'$. 
Hence $\gsta{\Gamma}{P/X}{\emptyset}{\Gamma'\sigma_M}$.

We now give a proof for $(2)$. We have to establish the following equivalence:  
$\netfir{\Gamma}{P/X}{\Gamma'}$ if and only if $\gsta{\Gamma}{P/X}{\emptyset}{\Gamma'}$. 
First, whenever $\gsta{\Gamma}{P/X}{\emptyset}{\Gamma'}$, then, by definition, $\netfir{\Gamma}{P/X}{\Gamma'}$. 
Conversely, $\netfir{\Gamma}{P/X}{\Gamma'}$ implies the existence of $\ell$ and $M$ such that $\lsta{\Gamma_{\ell}}{P/X}{M}{\Gamma^1_{\ell}}$. 
Thus, by definition $\gsta{\Gamma}{P/X}{M}{\Gamma^1}$ (with $\Gamma^1_{\ell'}=\Gamma_{\ell'}$ for $\ell' \not= \ell$). 
Hence, by $(1)$ we have $\gsta{\Gamma}{P/X}{\emptyset}{\Gamma^1\sigma_{M}}$. 
Note that productions applications are deterministic, and messages consumption too. 
Hence, it suffices to prove $\Gamma':= \Gamma^1\sigma_{M}$ to obtain the desired result.
In fact the nodes replacement performed to obtain $\Gamma'$ and $\Gamma^1$ are identical, 
since the same production is applied at the same node.
Let us denote by $\Gamma_[P/X]$ the configuration obtained by replacement of $X$. 
We have to show the equality $\Gamma'=\Gamma_[P/X]\sigma = \Gamma_[P/X] \sigma_l \sigma_M$, 
where $\sigma$ is the usual substitution applied during production application, $\sigma_l$ 
is the substitution resulting from applying production locally to $\Gamma_l$, and $\sigma_M$ 
is the substitution obtained by consumption of messages in $M$. Now, one can notice that all 
substitutions in $\sigma_M$ replace a variable $y$ by a term $t$ whenever $y$ is a subscription 
to some value produced in $\Gamma_l$. The effect is exactly the same as applying $\sigma$ at 
nodes that differ from $X$ in $\Gamma$. As additional subscription generated by messages 
consumption is not considered in the product, we have $\Gamma'=\Gamma_[P/X]\sigma = \Gamma_[P/X] \sigma_l \sigma_M$. 

The last statement, $(3)$, expresses the commutativity of the following diagram:
\begin{displaymath}
    \xymatrix{ & & \Gamma \ar[dl]^{M_1}_{P_1/X_1}\ar[dr]^{P_2/X_2}_{M_2} &  & \\
               \Gamma_1\sigma_{M_1} \ar[d]^{P_2/X_2}_{M_2} & \Gamma_1 \ar@{~>}[l]\ar[dr]^{P_2/X_2}_{M_2} &  
               & \Gamma_2\ar@{~>}[r]\ar[dl]^{M_1}_{P_1/X_1} & \Gamma_2\sigma_{M_2}\ar[d]^{P_1/X_1}_{M_1}\\
               \Gamma_1' \ar@{~>}[drr] &  & \Gamma_3\ar@{~>}[d] & & \Gamma_2'\ar@{~>}[dll]\\
                &   & \testlabel&  & }
\end{displaymath}

We first consider the commutativity of the center:
\begin{center}
$\gsta{\Gamma}{P_1/X_1}{M_1}{\Gamma_1}\gsta{}{P_2/X_2}{M_2}{\Gamma_3}$ commutes into 
$\gsta{\Gamma}{P_2/X_2}{M_2}{\Gamma_2}\gsta{}{P_1/X_1}{M_1}{\Gamma_3}$. 
\end{center}

This follows directly from the properties of input-enabled grammars: $\Gamma_3$ is simply $\Gamma$ 
where both open nodes $X_1$ and $X_2$ have been replaced respectively by the closed nodes 
$X_1=P_1(Y^1_1, \ldots, Y^1_m)$ and $X_2=P_2(Y^2_1, \ldots, Y^2_k)$ (since $M_1$ and $M_2$ 
have no common variables they are unaffected by each other).
 
Let us consider the left hand-side of the diagram.
From $(1)$, we have that $\xymatrix{\Gamma \ar[r]_{M_1}^{P_1/X_1} &\Gamma_1 \ar@{~>}[r]&\Gamma_1\sigma_{M_1}}$. 
And by (2), this implies that 
\netfir{\Gamma}{P_1/X_1}{\Gamma_1\sigma_{M_1}}. Using  Proposition~\ref{prop:monotony}, and the fact that 
\netfir{\Gamma_1}{P_2/X_2}{\Gamma_1}, whe have that $P_2$ is triggered and enabled in $\Gamma_1\sigma_{M_1}$. 
Hence, $\gsta{\Gamma_1\sigma_{M_1}}{P_2/X_2}{M_2}{\Gamma'_1}$. Using again 
Proposition~\ref{prop:monotony}, configuration $\Gamma'_1$ is simply $\Gamma_3\sigma_{M_1}$. 

Furthermore, from (1), we have that:
$\xymatrix{\Gamma_1\sigma_{M_1} \ar[r]_{M_2}^{P_2/X_2} &\Gamma'_1 \ar@{~>}[r]&\Gamma'_1\sigma_{M_2}}$, since 
$\Gamma'_1=\Gamma_3\sigma_{M_1}$, we have:  $\Gamma'_1\sigma_{M_2} = \Gamma_3\sigma_{M_1}\sigma_{M_2} $. 
By a symmetric argument, we obtain: $\Gamma'_2\sigma_{M_1} = \Gamma_3\sigma_{M_2}\sigma_{M_1} $.
Since these substitutions have disjoint support, we have 
$\sigma_{M_1}\sigma_{M_2} = \sigma_{M_2}\sigma_{M_1} =\sigma_{M_1}\cup\sigma_{M_2}$. Thus, we have: 
$\Gamma_3\sigma_{M_1}\sigma_{M_2} = \Gamma_3\sigma_{M_2}\sigma_{M_1} = \Gamma_3\sigma_{M_1}\cup\sigma_{M_2}$.

Finally $\xymatrix{\Gamma_3 \ar@{~>}[r]&\Gamma_3 (\sigma_{M_1}\cup\sigma_{M_2})}$ 
follows from the definitions of $\sigma_{M_1}$ and $\sigma_{M_2}$.

\PropBox{prop:distribution}
\end{proof}

%% file: soundness.tex
A specification is sound if every case can reach completion  no matter how its execution started. 
\begin{definition}\label{def:soundness}
Let a guarded attribute grammar be given with its interface (Definition~\ref{def:interface}). 
A \textbf{case} is an instanciation $c=s(t_1\sigma,\ldots,t_n\sigma)\langle x_1,\ldots,x_m\rangle$ of a service $s(t_1,\ldots,t_n)\langle x_1,\ldots,x_m\rangle$ 
where $\sigma$ is a substitution such that $\mathrm{var}(t_i)\sigma\subseteq\set{x_1,\ldots,x_m}$. Stated otherwise a case is, but for the variables with a synthesized value, 
a closed instance of a service. It means that it is a service call which already contains all the information coming from the environment of the guarded attribute grammar. 
A configuration is \textbf{closed} if it contains only closed nodes. 
A guarded attribute grammar is \textbf{sound} if a closed configuration is accessible from any configuration $\Gamma$ accessible from the 
initial configuration $\Gamma_0(c)=\set{X_0=c}$ associated with a case $c$. 
\ExpleBox{def:soundness}
\end{definition}
Let $\gamma$ denote the set of configurations accessible from the initial configuration of some case.
We consider the finite sequences $(\Gamma_i)_{0< i\leq n}$ and the infinite sequences $(\Gamma_i)_{0< i<\omega}$ of configurations in $\gamma$ 
such that $\netfir{\Gamma_i}{\,}{\Gamma_{i+1}}$. A finite and maximal sequence is said to be \textbf{terminal}, i.e.,\ a terminal sequence leads to a configuration 
that enables no production. 
Soundness can the be rephrased by the two following conditions. 
\begin{enumerate}
 \item Every terminal sequence leads to a closed configuration. 
 \item Every configuration on an infinite sequence also belongs to some terminal sequence. 
\end{enumerate}

Soundness can unfortunately be proved indecidable by a simple encoding of Minsky machines.

\begin{proposition}\label{prop:Minsky}
Soundness of guarded attribute grammar is undecidable.
\end{proposition}
\begin{proof}
We consider the following presentation of the Minsky machines.
We have two registers $r_1$ and $r_2$ holding integer values. 
Integers are encoded with the constant \textbf{zero} and the unary operator \textbf{succ}.
The machine is given by a finite list of instructions $\mathit{instr}_i$ for $i=1,\ldots,N$ 
of one of the three following forms
\begin{enumerate}
 \item \textbf{INC(r,i):} increment register $r$ and go to instruction $i$. 
 \item \textbf{JZDEC(r,i,j):} if the value of register $r$ is 0 then go to instruction $i$ else decrement the value of the register and go to $j$.
 \item \textbf{HALT:} terminate.
\end{enumerate}
We associate such a machine with a guarded attribute grammar whose sorts corresponds bijectively to the lines of the program,
(i.e., $S=\set{s_1,\ldots, s_N}$) with the following encoding of the program instructions by productions:
\begin{enumerate}
 \item If $\mathit{instr}_k = \mathrm{INC}(r_1,i)$ then add production 
 $$
  \mathrm{Inc}(k,1,i): s_k(x,y) \leftarrow s_{i}(\mathbf{succ}(x),y)
 $$
 \item If $\mathit{instr}_k = \mathrm{INC}(r_2,i)$ then add production 
 $$
  \mathrm{Inc}(k,2,i): s_k(x,y) \leftarrow s_{i}(x,\mathbf{succ}(y))
 $$
 \item If $\mathit{instr}_k = \mathrm{JZDEC}(r_1,i,j)$ then add the productions 
 \[
 \begin{array}{l@{\;:\;}l}
  \mathrm{Jz}(k,1,i) & s_k(\mathbf{zero},y) \leftarrow s_i(\mathbf{zero},y)\\
  \mathrm{Dec}(k,1,j) & s_k(\mathbf{succ}(x),y) \leftarrow s_j(x,y)
 \end{array}
 \]
 \item If $\mathit{instr}_k = \mathrm{JZDEC}(r_2,i,j)$ then add the productions 
 \[
 \begin{array}{l@{\;:\;}l}
  \mathrm{Jz}(k,2,i) & s_k(x,\mathbf{zero}) \leftarrow s_i(x,\mathbf{zero})\\
  \mathrm{Dec}(k,2,j) & s_k(x,\mathbf{succ}(y)) \leftarrow s_j(x,y)
 \end{array}
 \]
 \item If $\mathit{instr}_k = \mathrm{HALT}$ then add production 
 $$
  \mathrm{Halt}(k): s_k(x,y) \leftarrow 
 $$
\end{enumerate}
Since there is a unique maximal firing sequence from the initial configuration 
$\Gamma_0=\set{X_0=s_1(\mathbf{zero},\mathbf{zero})}$ 
the corresponding guarded attribute grammar is sound if and only if the computation 
of the corresponding Minsky machine terminates. 
\PropBox{prop:Minsky}
\end{proof}

%% file: conclusion.tex
To conclude we assess our model and highlight some research directions. 
\subsection{Assessment of the model}
In a nutshell {\em the model of GAGs provides a modular, declarative, user-centric, data-driven, distributed and reconfigurable model of case management}.  
\paragraph{Concurrency.} 
The lifecycle of a business artifact is implicitly represented by the grammar productions. 
A production decomposes a task into new subtasks and specifies constraints between their attributes in the form of the so-called semantic rules.  
The subtasks may then evolve independently as long as the semantic rules are satisfied. 
The order of execution, which may depend on value that are computed during process execution, need not (and cannot in general) be determined statically. 
For that reason, GAGs allow as much concurrency as needed.
In comparison, models in which the lifecycle of artifacts are represented by finite automata constrain concurrency among tasks in an artificial way. 

\paragraph{Modularity.} The GAG approach also facilitates a modular description of business processes. 
For instance when a referee has accepted to produce a report, one need not care about the subprocess dedicated to the actual production of the report.
In the example depicted on Table~\ref{Table:editorialProcess}, %provided in section~\ref{sec:AW}
 making a review report was modeled by a single production, 
but one can imagine similar situations where computing some synthesized information is given by a large 
set of rules used to collect and assemble information arising from various sources.  
However, following a top-down approach, one simply introduces an attribute in which this report should eventually be synthesized and 
delegate the actual production of the expected outcome to an additional set of rules. 
The identification of the different roles involved in the business process also contributes to enhance modularity. 
Finally, some techniques borrowed from attribute grammars, like descriptional composition \cite{Ganzinger83,GanzingerG84},  
decomposition by aspects \cite{WykMBK02,Wyk07} or higher-order attribute grammars \cite{VogtSK89}, %\cite{VogtSK89,SwierstraV91}, 
may also contribute to better modular designs.
  
\paragraph{Reconfiguration.}
The workflow can be reconfigured at run time: New business rules (productions of the grammar) can be added to the system without 
disturbing the current cases. 
By contrast, run time reconfiguration of workflows modeled by Petri nets (or similar models) is a complex issue~\cite{MichelisE96,EllisK00}.
One can also add ``macro productions'' corresponding to specific compositions of productions.
For instance if the Editor-in-chief wants to handle the evaluation of a paper, he can decide to act as an associate editor and as a referee for this particular submission.
However, this means forwarding the corresponding case to himself as an associate editor and then asking himself as a referee if he is willing to write a report. 
A more direct way to model this decision is to  encapsulate these steps in a compound macro production that bypasses the intermediate communications.   
More generally compound rules can be introduced for handling unusual behaviors that deviates from the nominal workflow. 

\paragraph{Logged information.} When a case is terminated, the corresponding artifact collects all relevant information of its history. 
Nodes are labeled by instances of the productions that have lead to the completion of the case. 
Henthforth, they record the decisions (the choices among the allowed productions) together with information associated with these decisions. 
In the case of the editorial process, a terminated case contains the names of the referees, the evaluation reports, the editorial decision, etc. 
A terminated case is a tree whose branches reflect causal dependencies among subactivities used to solve a case, 
while abstracting from concurrent subactivities. 
The temporary information stored  by the attributes attached to open nodes no longer exist when the case has reached completion. 
Closing nodes eliminates temporary information without resorting to any complex mechanism of distributed garbage collection.
The artifacts can be collected in a log which may be used for the purpose of process mining~\cite{Aalst11} either for process discovery 
(by inferring a GAG from a set of artifacts using common patterns in their tree structure) or for conformance checking 
(by inspection of the logs produced during simulations of a model or executions of an actual implementation). 

\paragraph{Distribution.} Guarded attributed grammars can easily be implemented on a distributed architecture without complex communication mechanisms (shared memory, FIFO channels,...). 
Stakeholders in a business process own open nodes, and communicate asynchronously with other stakeholders via messages. 
Moreover there are no edition conflicts since each part of an artifact is edited by the unique owner of the corresponding node. 

\subsection{Further works.}
We plan to design prototypes to analyze and implement a GAG description together with the required support tools 
(editor, parser, checker, simulators ...) and to concentrate on the following research directions:

\paragraph{Applicability.} 
We intend to develop some representative case studies to check  applicability and limitations of the model.
The first case study is a (simplified) distributed distance learning system, for which a GAGs implementation may have several advantages w.r.t. traditional solutions. 
First, since it does not rely on a client/server architecture, it should behave better in a degraded environment (when Internet connection is not always available).  
Most of the activity of the stakeholders are offline and the communication between them takes place upon avaibility of Internet connection. 
Second, the declarative decomposition of learning activities which does not impose a particular execution order together with the modularity of the model should provide more flexibility in the description of the learning processes.  
The second case study is a reporting system (e.g. semi-automatic synthesis of dashboards). 
The grammar can reflect the structure of the report, the identification of the stakeholders and their respective contributions lead to a distributed version of the grammar.
Finally the semantic rules implement the automatic assembly of the reports. 
We obtain a ``write things once'' principle: once data are collected in some synthesized position they can be used wherever they are needed.
GAGs should hence reduce the workload of stakeholders: as the largest part of activity reports collect information that is already available somewhere, and can be extracted automatically by a GAG. 

\paragraph{Structuring the workspace of a stakeholder: Active Workspaces.}
One can rely on the technique presented in  \cite{SwierstraAS98} 
to extract a domain specific language from the subgrammar associated with a stakeholder according to the role(s) he plays in the system. 
The grammar may then contribute to structure the workspace of a stakeholder:   
Procedural parts of the artifacts given by the semantic rules encode and encapsulate technical know-hows that the end user may safely 
ignore. Each stakeholder manipulates documents in a familiar syntax using notations adapted to his domain of expertise (the DSL derived 
from the semantic rules). This can be important for enabling end-users with low level computer literacy to take part in the business process. 

\paragraph{Soundness.} Soundness is a crucial issue of case management systems: it guarantees that a case has a way to reach completion 
from any accessible configuration. Unsurprisingly soundness of GAGs in undecidable. We are looking for recursive subclasses of 
GAGs with decidable soundness.

%% file: RR-8528.bbl
\begin{thebibliography}{10}

\bibitem{AbiteboulBMMW02}
Serge Abiteboul, Omar Benjelloun, Ioana Manolescu, Tova Milo, and Roger Weber.
\newblock Active xml: A data-centric perspective on web services.
\newblock In {\em BDA'02}, 2002.

\bibitem{Backhouse02}
Kevin Backhouse.
\newblock A functional semantics of attribute grammars.
\newblock In {\em Tools and Algorithms for the Construction and Analysis of
  Systems, TACAS}, volume 2280 of {\em Lecture Notes in Computer Science},
  pages 142--157. Springer, 2002.

\bibitem{ChiricaM79}
Laurian~M. Chirica and David~F. Martin.
\newblock An order-algebraic definition of knuthian semantics.
\newblock {\em Mathematical Systems Theory}, 13:1--27, 1979.

\bibitem{CourcelleF82}
Bruno Courcelle and Paul Franchi-Zannettacci.
\newblock Attribute grammars and recursive program schemes i and ii.
\newblock {\em Theor. Comput. Sci.}, 17:163--191 and 235--257, 1982.

\bibitem{DamaggioDV12}
Elio Damaggio, Alin Deutsch, and Victor Vianu.
\newblock Artifact systems with data dependencies and arithmetic.
\newblock {\em ACM Trans. Database Syst.}, 37(3):22, 2012.

\bibitem{DamaggioHV13}
Elio Damaggio, Richard Hull, and Roman Vacul\'{\i}n.
\newblock On the equivalence of incremental and fixpoint semantics for business
  artifacts with guard-stage-milestone lifecycles.
\newblock {\em Inf. Syst.}, 38(4):561--584, 2013.

\bibitem{DeransartM85}
Pierre Deransart and Jan Maluszynski.
\newblock Relating logic programs and attribute grammars.
\newblock {\em J. Log. Program.}, 2(2):119--155, 1985.

\bibitem{DeransartMBook93}
Pierre Deransart and Jan Maluszynski.
\newblock {\em A grammatical view of logic programming}.
\newblock MIT Press, 1993.

\bibitem{EllisK00}
Clarence~A. Ellis and Karim Keddara.
\newblock Ml-dews: Modeling language to support dynamic evolution within
  workflow systems.
\newblock {\em Computer Supported Cooperative Work}, 9(3/4):293--333, 2000.

\bibitem{Ganzinger83}
Harald Ganzinger.
\newblock Increasing modularity and language-independency in automatically
  generated compilers.
\newblock {\em Sci. Comput. Program.}, 3(3):223--278, 1983.

\bibitem{GanzingerG84}
Harald Ganzinger and Robert Giegerich.
\newblock Attribute coupled grammars.
\newblock In {\em SIGPLAN Symposium on Compiler Construction}, pages 157--170.
  ACM, 1984.

\bibitem{Hull08}
Richard Hull.
\newblock Artifact-centric business process models: Brief survey of research
  results and challenges.
\newblock In {\em OTM 2008}, volume 5332 of {\em Lecture Notes in Computer
  Science}, pages 1152--1163. Springer, 2008.

\bibitem{HullDMFGHHHLMNSV11}
Richard Hull, Elio Damaggio, Riccardo~De Masellis, Fabiana Fournier, Manmohan
  Gupta, Fenno~Terry Heath, Stacy Hobson, Mark~H. Linehan, Sridhar Maradugu,
  Anil Nigam, Piyawadee~Noi Sukaviriya, and Roman Vacul\'{\i}n.
\newblock Business artifacts with guard-stage-milestone lifecycles: managing
  artifact interactions with conditions and events.
\newblock In {\em Fifth ACM International Conference on Distributed Event-Based
  Systems, DEBS 2011}, pages 51--62. ACM, 2011.

\bibitem{Johnsson87}
Thomas Johnsson.
\newblock Attribute grammars as a functional programming paradigm.
\newblock In {\em Functional Programming Languages and Computer Architecture,
  FPCA}, volume 274 of {\em Lecture Notes in Computer Science}, pages 154--173.
  Springer, 1987.

\bibitem{KlintLV05}
Paul Klint, Ralf L{\"a}mmel, and Chris Verhoef.
\newblock Toward an engineering discipline for grammarware.
\newblock {\em ACM Transaction on Software Engineering Methodologies},
  14(3):331--380, 2005.

\bibitem{Knuth68}
Donald~E. Knuth.
\newblock Semantics of context free languages.
\newblock {\em Mathematical System Theory}, 2(2):127--145, 1968.

\bibitem{LohmannW10}
Niels Lohmann and Karsten Wolf.
\newblock Artifact-centric choreographies.
\newblock In {\em Service-Oriented Computing - 8th International Conference,
  ICSOC 2010, San Francisco, CA, USA, December 7-10, 2010.}, pages 32--46,
  2010.

\bibitem{MartelliM82}
Alberto Martelli and Ugo Montanari.
\newblock An efficient unification algorithm.
\newblock {\em ACM Trans. Program. Lang. Syst.}, 4(2):258--282, 1982.

\bibitem{Mayoh81}
Brian~H. Mayoh.
\newblock Attribute grammars and mathematical semantics.
\newblock {\em SIAM J. Comput.}, 10(3):503--518, 1981.

\bibitem{MichelisE96}
Giorgio~De Michelis and Clarence~A. Ellis.
\newblock Computer supported cooperative work and {P}etri nets.
\newblock In {\em Advanced Course on Petri Nets, Dagstuhl 1996}, volume 1492 of
  {\em Lecture Notes in Computer Science}, pages 125--153. Springer, 1998.

\bibitem{NigamC03}
A.~Nigam and N.~S. Caswell.
\newblock Business artifacts: An approach to operational specification.
\newblock {\em IBM Syst. J.}, 42:428--445, July 2003.

\bibitem{BPEL07}
OASIS.
\newblock Web services business process execution language.
\newblock Technical report, OASIS, 2007.
\newblock \url{http://docs.oasis-open.org/wsbpel/2.0/OS/wsbpel-v2.0-OS.pdf}.

\bibitem{Paakki95}
Jukka Paakki.
\newblock Attribute grammar paradigms - a high-level methodology in language
  implementation.
\newblock {\em ACM Computing Surveys}, 27(2):196--255, 1995.

\bibitem{SaraivaS03}
Jo{\~a}o Saraiva and S.~Doaitse Swierstra.
\newblock Generating spreadsheet-like tools from strong attribute grammars.
\newblock In {\em Generative Programming and Component Engineering, GPCE 2003},
  volume 2830 of {\em Lecture Notes in Computer Science}, pages 307--323.
  Springer, 2003.

\bibitem{SaraivaSK00}
Jo{\~a}o Saraiva, S.~Doaitse Swierstra, and Matthijs~F. Kuiper.
\newblock Functional incremental attribute evaluation.
\newblock In {\em Compiler Construction, CC 2000}, volume 1781 of {\em Lecture
  Notes in Computer Science}, pages 279--294. Springer, 2000.

\bibitem{SwierstraAS98}
S.~Doaitse Swierstra, Pablo R.~Azero Alcocer, and Jo{\~a}o Saraiva.
\newblock Designing and implementing combinator languages.
\newblock In {\em Advanced Functional Programming}, pages 150--206, 1998.

\bibitem{Aalst11}
Wil M.~P. van~der Aalst.
\newblock {\em Process Mining - Discovery, Conformance and Enhancement of
  Business Processes}.
\newblock Springer, 2011.

\bibitem{VogtSK89}
Harald Vogt, S.~Doaitse Swierstra, and Matthijs~F. Kuiper.
\newblock Higher-order attribute grammars.
\newblock In {\em PLDI}, pages 131--145, 1989.

\bibitem{Wyk07}
Eric~Van Wyk.
\newblock Implementing aspect-oriented programming constructs as modular
  language extensions.
\newblock {\em Sci. Comput. Program.}, 68(1):38--61, 2007.

\bibitem{WykMBK02}
Eric~Van Wyk, Oege de~Moor, Kevin Backhouse, and Paul Kwiatkowski.
\newblock Forwarding in attribute grammars for modular language design.
\newblock In {\em Compiler Construction, ETAPS 2002, Grenoble, France}, pages
  128--142, 2002.

\end{thebibliography}
